\documentclass[12pt,a4paper]{article}
\usepackage[dvipsnames]{pstricks} 
\usepackage{afterpage}
\usepackage{epsfig}
\usepackage{pst-grad} 
\usepackage{pst-plot} 
\usepackage{float}
\usepackage[stable]{footmisc}
\usepackage[utf8]{inputenc}
\usepackage{a4wide}
\usepackage{amsmath,amstext,amsthm,amssymb}
\usepackage{amsfonts}
\usepackage{framed}
\usepackage{graphicx}
\usepackage{color}
\usepackage{bbold}
\usepackage{bbm}
\usepackage{rotating} 
\usepackage{slashed} 
\usepackage{cancel} 
\usepackage{braket} 
\usepackage{array}
\usepackage{setspace}
\usepackage{hyperref}
\usepackage{fontawesome5}
\usepackage{overpic}
\usepackage[nosort]{cite}
\usepackage{lipsum}
\usepackage{dsfont}
\usepackage{noindentafter}
\usepackage{simpleicons}
\newtheoremstyle{mydefinition}
  {\topsep}   
  {0.5cm}     
  {\normalfont}
  {}          
  {\bfseries} 
  {:}         
  {5pt plus 1pt minus 1pt} 
  {\thmname{#1}\thmnumber{ #2}\thmnote{ \normalfont{(#3)}}} 
\theoremstyle{mydefinition}
\newtheorem{theorem}{Theorem}
\newtheorem{lemma}{Lemma}
\newtheorem{definition}{Definition}
\NoIndentAfterEnv{theorem}
\NoIndentAfterEnv{lemma}
\NoIndentAfterEnv{definition}

\newcommand{\del}[0]{\partial}

\let\baraccent=\=
\renewcommand{\=}[1]{\stackrel{#1}{=}}
\newcommand{\id}[0]{\mathds{1}}

\newcommand{\R}{\ensuremath{\mathbb{R}}}
\newcommand{\Z}{\ensuremath{\mathbb{Z}}}
\newcommand{\C}{\ensuremath{\mathbb{C}}}
\newcommand{\N}{\ensuremath{\mathbb{N}}}
\newcommand{\Proj}{\ensuremath{\mathbb{P}}}

\DeclareSymbolFontAlphabet{\mathbb}{AMSb}
\definecolor{JM}{HTML}{629202}
\definecolor{BH}{HTML}{f903d7}

\begin{document}

\pagestyle{plain}
\pagenumbering{gobble} 

\makeatletter
\@addtoreset{equation}{section}
\makeatother
\renewcommand{\theequation}{\thesection.\arabic{equation}}
\pagestyle{empty}
\vspace{0.5cm}

\begin{center}
    
	{\LARGE \bf{Calabi-Yau Orientifold Hypersurfaces and their F--theory Uplifts} \\[4mm]
    }

\end{center}

\vspace{1cm}

\begin{center}
	\scalebox{0.95}[0.95]{{\fontsize{14}{30}\selectfont Bj\"orn Hassfeld and Jakob Moritz}}
\end{center}

\begin{center}
	\vspace{0.25 cm}
	
		\textsl{Department of Physics, University of Wisconsin-Madison, Madison, WI 53706, USA}\\

	\vspace{1cm}
	\normalsize{\bf Abstract} \\[8mm]

\end{center}

\begin{center}
	\begin{minipage}[h]{15.0cm}
        
        We present an algorithm that constructs Calabi-Yau threefold orientifolds and their F--theory uplifts to elliptically-fibered Calabi-Yau fourfolds, embedded in toric varieties at codimension one and two respectively. The resulting Calabi-Yau fourfolds arise from triangulations of $6d$ reflexive polytopes --- which our method constructs from orientifold data --- and are smooth away from isolated terminal singularities. 
        For many of our fourfolds, the construction of the mirror manifold is immediate, enabling the computation of fourfold periods, and thus the seven-brane superpotential.
        We present multiple examples that demonstrate these capabilities. Our algorithms work with $\mathtt{CYTools}$ and are available through a GitHub repository \href{https://github.com/B-Hassfeld/Calabi-Yau-orientifolds-and-F-Theory-uplifts}{\faGithub}.
        
	\end{minipage}
\end{center}
\newpage
\setcounter{page}{1}
\pagestyle{plain}
\renewcommand{\thefootnote}{\arabic{footnote}}
\setcounter{footnote}{0}
%
%
\tableofcontents
\newpage
\pagenumbering{arabic} 
\setcounter{page}{1}

\section{Introduction}\label{sec:Introduction}

The unexplained hierarchy problems of the standard model of particle physics and cosmology beg for an embedding into a UV-complete framework of quantum gravity, i.e., string theory. Addressing in particular the nature of dark energy and the electroweak hierarchy problem requires understanding the distribution of local minima in the low-energy effective potential of string compactifications --- the string landscape \cite{Bousso:2000xa,Giddings:2001yu,Douglas:2003um,Ashok:2003gk,Denef:2004ze}.

One of the best-understood corners of the string landscape arises from compactifications of type IIB string theory on Calabi-Yau orientifolds with fluxes and branes \cite{Candelas:1985en,Gukov:1999ya,Dasgupta:1999ss,Giddings:2001yu,Kachru:2003aw,Grimm:2004uq,Balasubramanian:2005zx,Conlon:2005ki}. The resulting $4d$ effective field theory is $\mathcal{N}=1$ supergravity, and supersymmetry can plausibly be broken spontaneously at energies far below the Kaluza-Klein scale \cite{Kachru:2002gs}. 

In this framework, the quest for finding isolated minima of the effective potential --- called \emph{moduli stabilization} ---  requires a detailed understanding of the K\"ahler and superpotentials, at the classical level and including the leading (non-)perturbative quantum corrections, as in \cite{Kachru:2003aw,Balasubramanian:2005zx}. Thus, improving our understanding of the string landscape involves developing better methods to compute these. This work is aimed towards sharpening our understanding of the classical superpotential in flux compactifications of type IIB string theory.

Throughout, we will focus on Calabi-Yau threefold hypersurfaces in toric varieties \cite{Batyrev:1993oya,Kreuzer:2000xy}, and orientifolds thereof.\footnote{See the textbook \cite{cox2011toric} for an introduction into toric geometry.} For these, a large arsenal of computational methods has already been developed \cite{Kreuzer:2000xy,Braun:2017nhi,Demirtas:2018akl,Demirtas:2020dbm,Demirtas:2023als,Moritz:2023jdb,Dubey:2023dvu,MacFadden:2023cyf,MacFadden:2024him,MacFadden:2025ssx,MacFadden:2026hgm}, and explicit flux vacua with stabilized moduli have been constructed \cite{Demirtas:2019sip,Demirtas:2020ffz,Alvarez-Garcia:2020pxd,Demirtas:2021ote,McAllister:2024lnt}.

However, an essential piece is so far missing: the classical superpotential is determined by fluxes \cite{Giddings:2001yu}, and can, for special configurations of D7-branes, be computed in terms of the \emph{periods} of the underlying Calabi-Yau threefold \cite{Hosono:1993qy,Hosono:1994ax,Demirtas:2023als}. But, this approach does not naturally extend to study the superpotential at general points in the moduli space of D7-brane configurations (see however \cite{Braun:2008pz,Braun:2009bh,Arends:2014qca}). Hence, assessing the mass spectrum in cosmological vacua \cite{Demirtas:2021nlu}, and moreover the perturbative stability of non-supersymmetric de Sitter candidates \cite{McAllister:2024lnt} requires more powerful computational tools.

For this, we turn to F--theory \cite{Vafa:1996xn,Morrison:1996na,Morrison:1996pp,Denef:2008wq,Weigand:2018rez}: it
not only provides a non-perturbative extension, in the string coupling, of type IIB string theory, but it also allows for the use of geometric methods to study the dynamics of D7-branes in weakly coupled type IIB theory \cite{Sen:1996vd,Sen:1997gv}.
Furthermore, some of the most promising constructions of the standard model of particle physics in string theory have been constructed using F--theory, see e.g.~\cite{Cvetic:2019gnh,Raghuram:2019efb} as well as  \cite{Marchesano:2022qbx,Cvetic:2022fnv,Marchesano:2024gul} and refs.~therein.

The basic idea of F--theory is to interpret the spatially varying axio-dilaton profile in type IIB compactifications with D7-branes as a fibration of an auxiliary torus over the Calabi-Yau orientifold \cite{Vafa:1996xn, Morrison:1996na,Morrison:1996pp,Weigand:2018rez}. Unbroken $4d$ $\mathcal{N}=1$ supersymmetry requires that this fibration defines an elliptically fibered Calabi-Yau fourfold, and the deformation moduli of D7-branes are geometrized as complex structure moduli of the fourfold. Happily, the \emph{entire} classical superpotential --- including its dependence on the D7-brane moduli --- is then captured by the periods of the fourfold and a choice of quantized four-form flux \cite{Gukov:1999ya,Grimm:2009ef,Alim:2010za,Aganagic:2009jq,Honma:2017uzn,Cota:2017aal}.

In this paper, we will work out in detail the explicit constructions of the required elliptically fibered Calabi-Yau fourfolds of F--theory, starting from orientifolds of Calabi-Yau threefold hypersurfaces in toric varieties, in the spirit of \cite{Collinucci:2008zs,Collinucci:2009uh}. This process is known as \emph{uplifting} type IIB orientifolds to F--theory compactifications.

Our starting point is the classification of orientifolds of Calabi-Yau threefold hypersurfaces of \cite{Moritz:2023jdb}: given a Calabi-Yau threefold hypersurface $Y_3$ in a toric variety $V_4$, one can define a Calabi-Yau orientifold for every holomorphic involution of the toric variety $V_4$.
One first suitably restricts the defining polynomial of $Y_3$ and then considers the induced orbifold action on the hypersurface. 
This defines a Calabi-Yau orientifold hypersurface $B_3=Y_3/{\mathbb{Z}_2}$ that is cut out of an orbifold variety $\widetilde{V}_4:=V_4/{\mathbb{Z}_2}$ as a divisor $D_B$.

We improve on this setting in several ways: first, we restrict to cases where $\widetilde{V}_4$ is itself a toric variety.\footnote{An $\mathcal{O}(1)$ fraction of orientifold involutions remain after imposing this restriction.} This allows us to harness the power of toric geometry to study the orientifold hypersurface $B_3$.
Second, we develop a systematic toric blow up procedure that ensures smoothness of the resulting orientifolds. For this we select suitable refinements of the \emph{normal fan} of $D_B$, which is roughly analogous to Batyrev's construction \cite{batyrev1993dualpolyhedramirrorsymmetry} of smooth Calabi-Yau hypersurfaces in toric varieties.\footnote{For related combinatorial constructions of F-theory compactifications see \cite{Morrison:2012js,Morrison:2012np,Martini:2014iza,Taylor:2015isa,Halverson:2015jua,Taylor:2015ppa,Halverson:2017ffz,Taylor:2017yqr,Abbasi:2025lvn}.}

We then turn to uplifting this class of orientifolds to elliptically--fibered Calabi-Yau fourfolds. 
We do so by defining an appropriately twisted $\Proj_{[2,3,1]}$ fibration over $\widetilde{V}_4$ to form a total space $V_6^s$, as also considered in \cite{Collinucci:2008zs,Collinucci:2009uh}. 
The Weierstrass divisor $D_W$ therein together with the orientifold hypersurface $D_B$ then defines a complete-intersection-Calabi-Yau (CICY) $Y_4^s\subset V_6^s$. 
However, $Y_4^s$ is generically singular and we explicitly resolve the codimension-one singularities associated with \emph{non-Higgsable clusters} (NHCs) \cite{Morrison:2012np,Morrison:2012js,Morrison:2014lca,Halverson:2016vwx,Halverson:2017vde}, by performing appropriate toric blow-ups $V_6\to V_6^s$.
Smoothness of the Calabi-Yau fourfold $Y_4\subset V_6$ at higher codimension can, with suitable restrictions on the class of models, be achieved by again considering refinements of an appropriately defined normal fan. In this case, the $6d$ toric variety $V_6$ is defined by a fine, regular, and star triangulation (FRST) of a $6d$ reflexive polytope.
Our construction thus achieves sufficient smoothness of the F-theory uplift, and makes the mirror symmetry dictionary of  Batyrev-Borisov \cite{batyrev1993dualpolyhedramirrorsymmetry,borisov1993mirrorsymmetrycalabiyaucomplete,Batyrev:1994pg} directly applicable in many cases. 

As explained in \cite{Batyrev:1993oya,borisov1993mirrorsymmetrycalabiyaucomplete,Batyrev:1994ju,Batyrev:1994pg,Batyrev:1995ca,Batyrev:2007cq}, mirror symmetry can be understood in essentially combinatorial terms if the pair $(D_B,D_W)$ forms a \emph{nef-partition}. This in particular implies smoothness of $Y_4$ (away from terminal $\mathbb{Z}_2$ singularities at codimension four).
Not all our models satisfy this additional constraint. But for those that do, we can compute the Hodge numbers, and, 
more importantly, we have set the stage for explicit computations of the periods of $Y_4$. Thus one can begin to address the seven-brane moduli stabilization problem.

Finally, we construct explicit ensembles of Calabi-Yau orientifolds, and their F--theory uplifts, by applying our technology to the Kreuzer-Skarke list \cite{Kreuzer:2000xy}, incorporating all inequivalent $\Z_2$ actions of the 4d ambient toric varieties. 
Many such models, after applying our normal fan construction \S\ref{sec:Normal_fan}, yield nef-partitions of 6d reflexive polytopes. 
We find that $240,136,807$ out of the $473,800,776$ 4d reflexive polytopes admit at least one orientifold involution that leads to an F--theory uplift defined by a nef-partition. In total, we find $255,976,011$ nef-partitions generated from the Kreuzer-Skarke list.
These models are stored in a database made available here: \href{https://huggingface.co/datasets/jakobmoritz/F-theory_Uplifts_Nef-partitions}{\simpleicon{huggingface}}.

This paper is organized as follows: in \S\ref{sec:Orientifolds}, we  collect relevant facts about toric varieties. We then review the orientifold construction of \cite{Moritz:2023jdb}, and we explain how to deform models in order to ensure smoothness. 
In \S\ref{sec:Systematic_F_Theory_uplifts} and \S\ref{sec:Properties}, we construct  F--theory uplifts, ensure smoothness and embed our models into the Batyrev-Borisov construction of mirror symmetry \cite{borisov1993mirrorsymmetrycalabiyaucomplete,batyrev1994dualconesmirrorsymmetry}.
We present a variety of explicit examples in \S\ref{sec:Examples}, and present our ensemble of nef-partition uplifts.

\section{O3/O7 Calabi-Yau orientifolds in toric varieties}\label{sec:Orientifolds}
In this section, we review how to construct Calabi-Yau orientifolds as subvarieties of toric varieties where the orientifold involution is inherited from an involution of the ambient toric variety. 
In \S\ref{sec:toric_varieties}, we begin with a review of toric geometry and Calabi-Yau hypersurfaces therein.
In \S\ref{sec:toric_orbifolds} and \S\ref{sec:induced_action} we discuss, based on \cite{Moritz:2023jdb}, how to construct Calabi-Yau orientifolds as hypersurfaces in toric varieties. 
In \S\ref{sec:resolution} and \S\ref{sec:Normal_fan}, we develop procedures to find smooth orientifolds based on the construction of \S\ref{sec:toric_orbifolds}-\S\ref{sec:induced_action}.

\subsection{Toric varieties and their hypersurfaces}\label{sec:toric_varieties}
Let us begin by recalling the ingredients of toric geometry needed for the following construction.
We refer to \cite{cox2011toric} for a detailed treatment of the subject.
The reader familiar with toric geometry may skip ahead to \S\ref{sec:toric_orbifolds}.

Let $N$ be an $r$-dimensional lattice $N\simeq\Z^r$, and let $N_\mathbb{R}\simeq \mathbb{R}^r$ denote the vector space $N\otimes \mathbb{R}$. 
We denote by $M$ the dual lattice of $N$ and by $M_\mathbb{R}$ the dual vector space $M\otimes \mathbb{R}$.
The natural pairing between $N$ and $M$ is denoted by $\langle \cdot , \cdot \rangle :M\times N\to\Z$. 

\begin{definition}\label{def:Cone}
    A \emph{strongly convex, rational, polyhedral cone} $\sigma\subset N_{\mathbb{R}}$ is the set of all positive linear combinations of a fixed set of finitely many lattice points in $N$ such that $-\sigma\cap \sigma=\{0\}$. Throughout, we will denote by $\sigma^k$ strongly convex, rational, polyhedral cones of dimension $k$.
\end{definition}

\begin{definition}\label{def:toric_fan}
    A \emph{toric fan} $\Sigma$ is a set of strongly convex, rational, polyhedral cones $\sigma$ such that each face of a cone in $\Sigma$ is also in  $\Sigma$, and with every pairwise intersection of cones being a face of each.
\end{definition}

We denote by $\Sigma(d)\subset \Sigma$ the subset of $d$-dimensional cones of the fan $\Sigma$. One-dimensional cones are also called rays, and we denote by the set $\{v_i\}_{i=1}^n\subset N$ their integral primitive generators, where $n:=|\Sigma(1)|$. Given some cone $\sigma$, we denote by $\sigma(1)\subset \Sigma(1)$ its one-dimensional faces.
In practice, it is enough to specify the top-dimensional cones and define the fan as these cones and their lower-dimensional faces.

\begin{definition}\label{def:toric_variety}
    A (normal) \emph{toric variety} $V_{\Sigma}$ is defined by $\Sigma$ as follows. First, one associates a homogeneous coordinate $x_i\in\C$ to each of the rays generated by the $v_i$. We will also frequently denote by $x_v$ a homogeneous coordinate associated with a ray $v$ or by $x_i$ a homogeneous coordinate associated with a ray $v_i$.
\end{definition}

Next, one lets $Z$ be the union over all subvarieties
\begin{align}
    \{x_{i_1}=x_{i_2}=\ldots =x_{i_k}=0|\,  {v_{i_1},\ldots,v_{i_k}} \text{ are not contained in a common cone  of } \Sigma \}\, .
\end{align}
Then,
\begin{align}
V_{\Sigma}:=\frac{\C^{n}- Z}{G}\, ,
\end{align}
where $G$ is the group of toric scaling relations defined by
\begin{align}
    G:=\{\eta \in (\C/\Z)^{n} : \sum_{i=1}^{n}\eta^iv_i\in N\}\, ,
\end{align}
with action on the homogeneous coordinates $x_i$
\begin{equation}
    x_i\mapsto e^{2\pi i\eta^i}x_i\, .
\end{equation}

The component of the group of toric scaling relations connected to the identity is generated by $x_i\mapsto \lambda^{Q_{\alpha i}} x_i$, $\alpha=1,\ldots, h^{1,1}(V)$, in terms of a \emph{GLSM charge matrix $Q_{\alpha i}$}. Here, the rows of $Q_{\alpha i}$ form a $\mathbb{Z}$-basis of linear relations among the $v_i$,
\begin{align}
    \sum_i Q_{\alpha i}v_i=0\,,\qquad \forall \alpha\,.
\end{align}

Toric varieties deserve their name because of the existence of an algebraic torus $T=(\C^*)^r\subset V_{\Sigma}$ that acts on $V_{\Sigma}$. $T$ is naturally defined as $T=(N\otimes \mathbb{C})/N$. A torus representative
\begin{align}
    t=\sum_{i=1}^{n}\eta_t^iv_i \in N\otimes \mathbb{C}
\end{align}
acts on $V$ via
\begin{align}\label{eq:algebraic_torus_action}
    \phi_{[t]}: x_i\to e^{2\pi i \eta_t^i}x_i\, .
\end{align}
Due to the identification via the action of the group $G$, \eqref{eq:algebraic_torus_action} is well-defined, i.e. it only depends on $t$ modulo $N$.

\begin{definition}
    A toric variety $V_{\Sigma}$ is called \emph{simplicial} if all cones $\sigma\in \Sigma$ are simplicial, i.e., they are generated by a number of rays equal to their dimension.
\end{definition}

\begin{lemma}[\cite{cox2011toric} Thm.~3.1.19]
    A simplicial toric variety is smooth if and only if all $\sigma\in \Sigma$ are smooth, i.e., primitive generators of $\sigma$ generate the sub-lattice $N\cap \sigma$.
\end{lemma}

\begin{lemma}[\cite{cox2011toric} Thm.~3.1.19 \&  Thm.~3.4.6]
    A toric variety $V_{\Sigma}$ is compact if and only if $\Sigma$ is a complete fan, i.e., the union over all $\sigma\in \Sigma$ is all of $N_{\mathbb{R}}$. This is equivalent to $V_{\Sigma}$ being complete.
\end{lemma}

Each $k$-dimensional cone $\sigma^k$ of the fan $\Sigma$ defines a codimension-$k$ subspace of $V$
\begin{align}
    \mathcal{F}(\sigma^k)=\{x_{v}=0\,\, \forall v \,\, \text{that generate } \sigma^k\}\, .
\end{align}
The codimension-one subvarieties 
\begin{equation}
    D_i = \{x_i=0\}
\end{equation}
associated to the rays $v_i$ define the torus invariant, or \emph{prime toric divisors}.

\begin{definition}\label{def:Weil_and_Cartier_divisor}
    A torus-invariant \emph{Weil divisor} is a formal linear combination of prime toric divisors $D=\sum_i a_iD_i$, with integer coefficients $a_i$. Such a Weil divisor is a \emph{Cartier divisor} if
\begin{equation}\label{eq:Cartier_div}
    \forall \sigma^r\in \Sigma:\quad \exists m_\sigma\in M \,: \quad \langle m_\sigma, v_i\rangle+a_i=0 \quad \forall v_i\in \sigma^r(1)\, .
\end{equation}
The set $\{m_\sigma\}$ furnishes the \emph{Cartier data} of a Cartier divisor.
\end{definition}

Cartier divisors correspond to line bundles on $V_{\Sigma}$, while Weil divisors in general only correspond to sheaves. 

Local monomial sections of a sheaf $\mathcal{O}_{V_\Sigma}(D)$, with $D=\sum_i a_i D_i$, are given by characters $m\in M$,
\begin{equation}\label{eq:monomial_sections}
    s_m=\prod_{i=1}^n x_i^{\langle m, v_i\rangle+a_i}\, ,
\end{equation}
and we have $\mathcal{O}_{V_\Sigma}(D)=\mathcal{O}_{V_\Sigma}(D')$ if and only if $D-D'$ is a \emph{principal divisor}, i.e., $D-D'=\sum_i \langle m,v_i \rangle D_i$ for some character $m\in M$. This motivates the equivalence relation
\begin{equation}
    D\sim D' \quad \Leftrightarrow \quad D-D' \,\, \text{is principal.}\label{eq:Equivalence_relation_principal_divisors}
\end{equation}

\begin{definition}\label{def:Divisor_class_group}
    The group of torus-invariant Weil divisors modulo the equivalence relation \eqref{eq:Equivalence_relation_principal_divisors} is called the \textit{divisor class group}.
\end{definition}

\begin{definition}\label{def:Newton_Polytope}
    The \emph{Newton Polytope} $\Delta_D$ of a torus-invariant Weil divisor $D$ is defined as 
    \begin{align}
        \Delta_D := \{m\in M_\R | \langle m,v_i\rangle \geq -a_i\,\, \forall i=1,\ldots,n\} \subset M_\R\, .
    \end{align}
\end{definition}

The vector space of global sections $\Gamma(\mathcal{O}_{V_\Sigma}(D))$ is generated by the globally defined monomial sections \eqref{eq:monomial_sections}, which in turn are defined from integral characters $m$ in the Newton Polytope $\Delta_D$.
Thus, the most general global section of $\mathcal{O}_{V_\Sigma}(D)$ takes the form
\begin{align}
    P=\sum_{m\in M\cap \Delta_D} \psi_m s_m\,,\qquad \psi_m\in\C\,,\label{eq:hypersurface_polynomial}
\end{align}
with the complex parameters $\psi_m$ parameterizing the family of hypersurfaces $\{P=0\}$.

\begin{definition}
    The \textit{Minkowski-sum} of two polytopes $\Delta_1,\Delta_2$ yields a new polytope
\begin{align}
    \Delta_1+\Delta_2 := \{p_1+p_2\,|\,p_1\in \Delta_1\,, p_2\in \Delta_2\}\, .
\end{align}
\end{definition}

\begin{definition}
    The \textit{basepoint locus} $BL_D$ of a divisor $D$ is defined as the subspace $BL_D\subset V$ contained in every member of the family $\{P=0\}$. Correspondingly, a divisor $D$ is \textit{basepoint free} if for all points $p\in V$, there exists a section $s$ of $D$ such that $s(p)\neq 0$.
\end{definition}

\begin{definition}\label{def:nef}
    A divisor is \textit{numerically effective} or just \textit{nef} for short, if the corresponding hypersurface has only non-negative intersections with effective curves, that is
\begin{align}
    D\cdot C \geq 0\,,\qquad \forall\text{ effective curves } C \text{ of }V\,.
\end{align}
\end{definition}

We collect the following statements about divisors:

\begin{lemma}[\cite{cox2011toric} Exercise 4.3.2]\label{lem:shift}
    The Newton Polytopes of two torus-invariant Weil divisors $D,D'$ that are equivalent under the relation \eqref{eq:Equivalence_relation_principal_divisors} differ by a translation of an element $m\in M$.
\end{lemma}

\begin{lemma}[\cite{cox2011toric} Theorem.~6.1.7]\label{lem:NP_Cartier_and_nef_is_lattice}
    A torus-invariant Cartier divisor $D$ is \emph{nef} if its Cartier data is contained in $\Delta_D$.
    If the fan $\Sigma$ of the underlying toric variety is complete, the Newton polytope $\Delta_D$ of a Cartier and nef divisor $D$ is a lattice polytope, that is, a polytope whose vertices are lattice points.
\end{lemma}

\begin{lemma}[\cite{cox2011toric} \S6]\label{lem:nef_Mink_sum}
    Given two torus-invariant divisors $D_1$ and $D_2$ that are Cartier and nef with associated Newton polytopes $\Delta_1$ and $\Delta_2$, the Newton polytope of $D_1+D_2$ is given by the Minkowski-sum $\Delta_1+\Delta_2$.
\end{lemma}

\begin{definition}
    The \emph{normal fan} $\Sigma_\Delta$ of a full-dimensional lattice polytope
\begin{align}
    \Delta=\{\langle m,v_i \rangle+a_i \geq 0\,\, \forall i\}\subset M_{\mathbb{R}}\label{eq:normal_fan_def}
\end{align}
is constructed by associating a full-dimensional cone $\sigma_m\subset N_{\mathbb{R}}$ to every vertex $m$ of $\Delta$. $\sigma_m$ is the cone over all the $v_i$ associated with hyperplanes $H_i:=\{\langle m,v_i \rangle+a_i= 0\}\subset M_{\mathbb{R}}$ that contain $m$.
\end{definition}

\begin{lemma}[\cite{cox2011toric} \S2.3]\label{lemma:normal_fan}
    Let $\Delta$ be a full-dimensional lattice polytope, and let $\Sigma$ be its normal fan. Further, let $V\equiv V_{\Sigma}$ and let $D$ be the torus-invariant divisor whose Cartier data are the vertices of $\Delta$. Then, $D$ is Cartier and ample, and in particular nef. For any refinement $\Sigma'$ of $\Sigma$, corresponding to a blow-up $\pi: \,V_{\Sigma'}\rightarrow V_{\Sigma}$ the pullback $D':=\pi^*(D)$ remains Cartier and nef. We have $D'=\sum_i a_i D_i$ with $a_i=-\langle m,v_i\rangle$ for $v_i$ contained in a cone $\sigma_m$ dual to a vertex $m\in \Delta$.
\end{lemma}

\begin{theorem}[\cite{cox2011toric} Prop.~6.1.13]\label{thm:normal_fan_of_Minkowski_sum}
    Let $\Delta^{(1)},\Delta^{(2)}$ be full dimensional lattice polytopes. Then, the normal fan $\Sigma_{\Delta^{(1)}+\Delta^{(2)}}$ is the coarsest common refinement of $\Sigma_{\Delta^{(1)}}$ and $\Sigma_{\Delta^{(2)}}$.
\end{theorem}

\begin{definition}\label{def:reflexive_and_polar_dual}
    A polytope $\Delta\subset M_\R$ is reflexive if $0\in \Delta$ and the polar dual polytope
\begin{align}
    \Delta^\circ := \{n\in N_\R: \langle m,n\rangle \geq -1\,,\,\forall m \in \Delta\}\,,
\end{align}
is a lattice polytope. In this case $(\Delta^\circ)^\circ = \Delta$.
\end{definition}

\begin{theorem}[\cite{cox2011toric} Thm~6.3.12]\label{thm:nef_is_bp_free}
    A torus-invariant Cartier and nef divisor in a complete and normal toric variety is basepoint free.
\end{theorem}

\begin{theorem}[Bertini, \cite{hartshorne1977algebraic} II 8.18 and III 7.9.1]\label{thm:Bertini}
    Let $Y\subset V$ be a subvariety defined by the vanishing of a sufficiently general global section of a line bundle $\mathcal{L}=\mathcal{O}_V(D)$, with $D$ basepoint free. Then, $Y$ is smooth away from singularities of $V$.
\end{theorem}

Every Weil divisor $D$ on a toric variety $V$ defines a family of codimension-one subspaces via Eq.~\eqref{eq:hypersurface_polynomial}. 
Calabi-Yau varieties can be defined considering anticanonical hypersurfaces $Y$ of $V$.
Their element in the divisor class group can be represented as
\begin{align}
   \overline{K}_V =\sum_i D_i\,.
\end{align}
The Newton polytope of the above torus-invariant divisor will be denoted simply by $\Delta$.
The adjunction theorem guarantees that the resulting hypersurfaces $Y$ have trivial canonical bundle and are hence Calabi-Yau.
The parameters $\psi_q$ in \eqref{eq:hypersurface_polynomial} (over-)parametrize the part of the complex structure moduli space of $Y$ that is inherited from $V$.

Mirror symmetry is best understood for the case when $\overline{K}_V$ is nef and Cartier \cite{Batyrev:1993oya}. 
In this case, its Newton polytope $\Delta$ is reflexive and the toric fan $\Sigma$ derives from a fine, regular and star triangulation (FRST) of the polytope $\Delta^\circ$.
A triangulation is a prescription that associates to every nonzero point in $\Delta^\circ$ a ray of $\Sigma$ and picks a consistent way to form a fan out of these rays.
Now, as $\Delta$ and $\Delta^\circ$ form a reflexive pair, we can also triangulate the polytope $\Delta$, and obtain fans $\Sigma^\vee$ of a toric variety $V^\vee$.
The anticanonical hypersurfaces $Y^\vee$ in toric fourfolds $V^\vee$ of this form are collectively mirror dual to the anticanonical hypersurfaces $Y$ in fourfolds $V$ \cite{Batyrev:1993oya}.

\begin{lemma}[\cite{Batyrev:1993oya}]\label{lem:points_int_facets_do_not_intersect}
    Let $\Sigma$ be a fan of a toric variety that is derived from an FRST of a reflexive polytope $\Delta$. Let $Y$ be an associated Calabi-Yau manifold. Then, the toric divisors $D_i$ associated to points interior to facets (codimension-one faces) of $\Delta$ do not intersect $Y$. 
\end{lemma}

As a consequence of Lemma \ref{lem:points_int_facets_do_not_intersect}, we only have to consider the points not interior to facets when studying Calabi-Yaus derived from FRSTs of reflexive polytopes.

 \paragraph{Remark:} 
The Weil divisors (not merely the torus-invariant ones) modulo linear equivalence actually form the same divisor class group as the one defined in Def.~\ref{def:Divisor_class_group}, see \S$4.1$ of \cite{cox2011toric}.
The Cartier and nef properties depend only on the linear-equivalence class of a divisor. Likewise, linearly equivalent divisors determine the same spaces of global sections and hence the same base locus. 
Thus from now on, we will use the same symbol, such as $D$, for an element of the divisor class group and for a divisor representing that class, hoping that the intended meaning is clear from the context.

\subsection{Toric orbifolds}\label{sec:toric_orbifolds}
We now describe how to construct orbifolds of toric varieties.
This is useful in order to study the orbifold action on the corresponding Calabi-Yau hypersurfaces.

For every discrete subgroup $G$ of the automorphism group $Aut(V,\C)$ of $V$, one may define a corresponding orbifold $\widetilde{V}=V/G$, which is a (generally singular) toric variety itself. The starting point is the automorphism group of the toric variety $V$, which has been understood by Cox \cite{cox92}. Moreover, because the orbifold $\widetilde{V}$ only depends on $G$ modulo conjugation by elements of $Aut(V,\C)$, one is actually interested in conjugacy classes of subgroups of $Aut(V,\C)$.

Conjugacy classes of $\mathbb{Z}_2$ subgroups in $Aut(V,\C)$ were studied in \cite{Moritz:2023jdb}: they are generated by a map $g$, which can be represented as
\begin{align}
    g=g_{\rm perm}\circ \phi_{[\xi/2]}\,,
\end{align}
where $\phi_{[\xi/2]}$ is an algebraic torus action
\begin{align}
    x_i\to e^{\pi i \xi^i}x_i\,,\label{eq:algebraic_action}
\end{align}
defined by a lattice vector $\xi\in N$ having the representation $\xi = \sum_i \xi^iv_i$, \emph{cf.} \eqref{eq:algebraic_torus_action}, and $g_{\rm perm}$ is a lattice automorphism of $N$ permuting the rays $\Sigma(1)$ such that the fan $\Sigma$ stays invariant.
Throughout this paper, we will restrict to involutions for which the permutation part is trivial, i.e., $g_{\text{perm}}=\id$.

The simplest such examples are given by involutions that act as a reflection
\begin{align}
    x_{i_0}\to -x_{i_0}\label{eq:simple_reflection}
\end{align}
on one of the homogeneous coordinates. 
This involution is induced by the choice $\xi=v_{i_0}\in N$.

Clearly, shifts of the $\xi^i$ by elements in $2 N$ leave the algebraic action \eqref{eq:algebraic_action} unchanged, such that the space of distinct algebraic involutions is parametrized by the coset
\begin{align}
    \frac{N}{2N}\simeq \left(\Z_2\right)^r\,.
\end{align}
The toric orbifold $\widetilde{V}=V/{\mathbb{Z}_2}$ is now constructed by refining the lattice $N$ as follows. One simply adds the lattice shifted by the half-integral element $\xi/2$ to form a new lattice $\tilde N:=N\cup (N+\xi/2)$. 
The toric fan of the orbifold variety $\Tilde{\Sigma}$  is then given by the old fan in the new lattice. In particular, if $V$ was a simplicial toric variety, then so is $\widetilde{V}$.

For later purposes it will be useful to note that, after the lattice refinement, some ray generators $v$ may no longer be primitive, if viewed as elements of the refined lattice $\widetilde{N}$. Instead, the primitive generators associated with the one-dimensional cones of the fan $\Tilde{\Sigma}$ may differ from those of $\Sigma$ by factors of $2$.
For the remainder of this section, we will denote the primitive rays in the new fan as $\tilde v_i$ and the associated homogeneous coordinates as $\tilde{x}_i$. If $\tilde{v}_i=v_i/2$ for some $i$, then the corresponding homogeneous coordinates are naturally identified as $\tilde{x}_i=x_i^2$.

It will be important to understand the fixed point locus of the orbifold action, as this gives rise to the loci hosting orientifold planes. The fixed locus turns out to be stratified by the torus invariant subvarieties defined by the cones $\sigma^k$ that satisfy \cite{Moritz:2023jdb}
\begin{align}
    \xi + \sum_{v \text{ generating } \sigma^k}v \in 2N\,.\label{eq:cones_fixed_locus}
\end{align}
Since by construction $\xi \in 2\widetilde{N}$, for every cone associated with a stratum of the fixed locus, i.e.,  satisfying \eqref{eq:cones_fixed_locus}, we have the relation
\begin{align}
    \sum_{ v_i \text{ generating } \sigma^k} v_i \in 2\widetilde{N}\, .\label{eq:cones_fixed_orbifold}
\end{align}
Thus, the primitive generators of one-cones (i.e., rays) that correspond to codimension-one fixed locus components of the orbifold action are even elements of $\widetilde{N}$.
Therefore, the corresponding rays of $\widetilde{\Sigma}$ have primitive generators $\tilde{v}_i=v_i/2\in \widetilde{N}$.

\subsection{Induced action on Calabi-Yau hypersurfaces}\label{sec:induced_action}
In order for the involution $\phi_{[\xi/2]}$ to descend to an involution of a Calabi-Yau hypersurface $Y\subset V$ defined by the polynomial equation $P_Y=0$, the polynomial $P_Y$ needs to get mapped to itself up to an overall $\mathbb{C}^*$ factor,
\begin{align}
    P_Y\mapsto \gamma_Y\cdot P_Y\,,\qquad \gamma_Y\in \C^*\,.    
\end{align}
This is usually not the case for the most general polynomial of the form \eqref{eq:hypersurface_polynomial}. Rather, to define a sensible involution of the hypersurface, one needs to first restrict the vector space of global sections to a special sub-locus where some of the parameters $\psi_m$ vanish.

Defining a parameter $\lambda_Y\in\{0,1\}$ labeling the two classes of the $\Z_2$ involution (which correspond to orientifolds of O5/O9 or O3/O7 type respectively), the parameters $\psi_q$ that remain unconstrained (are ``projected in'') correspond to characters $q\in M$ that satisfy
\begin{align}
    \langle q,\xi\rangle +\lambda_Y = 0 \mod 2\, .\label{eq:projected_in}
\end{align}
As we are interested in O3/O7 orientifolds in this work, we will henceforth only consider the case $\lambda_Y=1$, and we call the orientifold covariant polynomial $P_B$.\footnote{We here preemptively choose the subscript $B$ for the orientifold, as we later want to consider such orientifolds as the ``bases'' of the elliptically fibered Calabi-Yau fourfolds that correspond to the F--theory uplifts.} We may now view it as a \emph{general} section of a sheaf $\mathcal{O}_{\widetilde{V}}(D_B)$, where $D_B$ is the divisor class group element corresponding to the induced hypersurface in $\widetilde{V}$ (which will no longer be anticanonical). We will refer to the hypersurface $\{P_B=0\}\subset \widetilde{V}$ as $B$.

Importantly, due to the need for tuning of $P$, even the generic orbifold invariant hypersurface in $V$ may no longer be guaranteed to be smooth. Or, from the ``downstairs'' perspective, $D_B$ may fail to be nef and may even fail to be Cartier.

The class of the O7 planes, i.e., the divisorial fixed locus, is in general given by
\begin{align}
    D_{O7}=2\overline{K}_{B} = 2(\overline{K}_{\widetilde{V}}-D_B)|_{B}\,,\label{eq:O7_class}
\end{align}
where we have used the adjunction theorem again.

\subsection{Some possible singularities}\label{sec:resolution}

As we just discussed, the tuning of the Calabi-Yau defining equation \eqref{eq:projected_in} may render the generic $\mathbb{Z}_2$-invariant hypersurface, and also its orbifold $B$, singular. In some cases the singularities are quite severe (they occur at infinite distance in complex structure moduli space), leading us to discard the corresponding model altogether. In other cases the tuning \eqref{eq:projected_in} leads to singularities (at finite distance in moduli space) that can be smoothed via a blow-up morphism that amounts to an arbitrarily small move in the moduli space. In this case the model makes sense as a string theory target space.

Before dealing with singularities systematically, let us survey a few obvious ways that smoothness can fail after restricting to the monomials \eqref{eq:projected_in}.

To facilitate this discussion we first note that for each cone $\sigma^k$ satisfying \eqref{eq:cones_fixed_locus}, one has
\begin{align}
    P|_{\mathcal{F}(\sigma^k)}\equiv 0\qquad \text{if}\qquad \dim(\sigma^k)+\lambda_Y+1=k+\lambda_Y +1\in 2\N\, ,\label{eq:fixed_locus_dimension}
\end{align}
and otherwise, the cone $\sigma^k$ defines a codimension-$k$ subvariety of $Y$. As a consequence, for O3/O7 orientifolds (for which $\lambda_Y=1$) all even-dimensional strata of the $\mathbb{Z}_2$ fixed locus in $V$ lie inside the general $\mathbb{Z}_2$-invariant hypersurface. On the other hand, all odd-dimensional strata intersect the hypersurface transversely, and thus all strata of the fixed locus in $Y$ have even complex dimension.

We now turn to discuss possible singularities. They can arise already at codimension one of $B$, whenever the general $\mathbb{Z}_2$-covariant polynomial factorizes
\begin{align}
     P_B = \tilde{x}_i\cdot P_B'(\tilde{x})\, .
\end{align}
In other words, the divisor $D_B$ has a basepoint locus at codimension one, and is hence reducible. This is an example of a singularity at infinite distance in moduli space, which we will not attempt to resolve.

On the other hand, given a stratum of the fixed locus at codimension two in $V$, it follows from \eqref{eq:fixed_locus_dimension} that it lies entirely in the hypersurface. Hence, the hypersurface polynomial can be written as
\begin{align}
     P_B = \tilde{x}_i P_B^{(1)}(\tilde{x})+\tilde{x}_j P_B^{(2)}(\tilde{x})\, ,
\end{align}
with $P_B^{(1)}$ and $P_B^{(2)}$ polynomials in the $\tilde{x}$, and $\tilde{x}_i=\tilde{x}_j=0$ defining the codimension-two stratum of the fixed locus.

The above form of $P_B$ signals the presence of singularities, generically of conifold type, along the point set
\begin{align}
    \{\tilde x_i= P_B^{(1)}=P_B^{(2)}=\tilde{x}_j=0\}\subset \widetilde{V}\, ,
\end{align}
along which the Jacobian criterion is violated.

In the present case, this occurs if and only if the toric variety $\widetilde{V}$ has, as a consequence of \eqref{eq:cones_fixed_locus}, \eqref{eq:cones_fixed_orbifold}, $A_1$-type singularities along the subspace
\begin{align}
    \{\tilde{x}_i=0=\tilde{x}_j\}\subset \widetilde{V}\,.
\end{align}
Now, unlike the previous situation, the point-like singularities arise at finite distance in moduli space, and it is sensible to resolve them by blowing up the singular locus. Here, one simply adds an exceptional divisor $E$ associated with a new ray in the toric fan of $\widetilde{V}$, with primitive generator $e$ defined by $2e=\tilde v_i+\tilde v_j$. This removes the locus of $A_1$-singularities, adding a family of exceptional $\mathbb{P}^1$s and simultaneously resolves the conifold singularities of the underlying $\mathbb{Z}_2$-invariant Calabi-Yau hypersurface. From the perspective of the Calabi-Yau hypersurface this is a small resolution.

This step of course changes the topology of $\widetilde V$ and $B$ (as well as the topology of their double-covers). In particular, the Hodge numbers change.

The remaining singularities at codimension three are then of orbifold type. When considering Calabi-Yau threefolds,  these codimension-three singularities are interpreted as O3 planes.
These are physical, and remain singular from the perspective of F--theory. Thus, there is no need to remove such singularities for our following purposes.

\subsection{Smooth orientifolds from normal fans}\label{sec:Normal_fan}
While the above subsection deals with the simplest types of singularities that can arise in our setup, avoiding these simple cases still does not guarantee that the resulting hypersurfaces are smooth. Indeed, without additional properties, proving smoothness of a general hypersurface is difficult (though of course it can be done on a case-by-case basis). 

Instead, we now turn to a systematic procedure that \emph{modifies} the toric variety $\widetilde{V}$ with suitable blow-ups and restrictions of triangulations, such that the resulting orientifold is \emph{guaranteed} to be smooth. Our procedure applies whenever the following conditions are met:

\begin{enumerate}
    \item The general global section of $\mathcal{O}_{\widetilde{V}}(D_B)$ must not factorize. In particular, $D_B$ must be basepoint free at codimension one.
    \item The convex hull over the integral points in the Newton polytope of $D_B$ must have full dimension.
\end{enumerate}
These exclude being forced to certain loci at \emph{infinite} distance in moduli space, such as those where the Calabi-Yau hypersurface $Y$ and its orientifold $B$ become reducible.

\begin{definition}\label{def:regular_orientifold}
    An orientifold hypersurface $B$ that satisfies the above conditions (a) and (b) is called \emph{regular}, and \emph{irregular} otherwise.
\end{definition}

Given a regular orientifold, smoothness is still not guaranteed, and depends in particular on the details of the fans $\widetilde{\Sigma}$, not just its one-skeleton $\widetilde{\Sigma}(1)$. As different triangulations of vector configurations that keep the one-skeleton fixed induce finite distance moves in the moduli space of the Calabi-Yau hypersurface $Y$ --- see in particular \cite{MacFadden:2025ssx} --- one can aim to restrict the set of triangulations of the vector configuration $\Sigma(1)$ (supplemented by suitable blow-up divisors) to those in which the hypersurface $B$ is guaranteed to be smooth.

Crucially, smoothness of $B$ (away from orbifold singularities of $\widetilde{V}$) is guaranteed if $D_B$ is Cartier and nef, due to Theorems  \ref{thm:nef_is_bp_free} and \ref{thm:Bertini}. 
For divisors $D_B$ that satisfy the conditions (a) and (b) above, one can then systematically construct fans $\widetilde{\Sigma}$ such that $D_B$ (or rather its pullback under a suitable birational morphism) is indeed a Cartier and nef divisor.

In order to accomplish this, we now let $\Delta\subset M$ be the lattice polytope defined as the \emph{convex hull over the integral points of the} Newton polytope $\Delta_B$ of $D_B$. Any refinement of its normal fan $\Sigma_\Delta$ naturally defines a toric variety in which the pullback of $D_B$ is a nef Cartier divisor, \emph{cf.} Lemma \ref{lemma:normal_fan}.
For example, one might start with a toric variety $\widetilde{V}$ with fan $\widetilde{\Sigma}$ and a divisor $D_B$ representing a Calabi-Yau orientifold.
One then computes its normal fan $\Sigma_\Delta$ to which one may add as many new rays as desired. In particular, one may add the rays of $\widetilde{\Sigma}$.

In this way, one may define a toric variety that is birational to $\widetilde{V}$, with the feature that the pullback of $D_B$ under the birational morphism is a Cartier and nef divisor.

A couple of comments are in order.
First, for the present purpose, not all the resulting toric hypersurfaces $B$ are acceptable: the generic section of the sheaf $\mathcal{O}_B(2\overline{K}_B)$ parameterizes the O7-plane, \emph{cf.} \eqref{eq:O7_class}, so it should not have generic vanishing orders larger than one.  Otherwise, one cannot interpret the geometry as the target space of a weakly coupled type IIB orientifold model.

Second, the process outlined  above may of course significantly alter the topology of the original toric variety $\widetilde{V}$. However, it only amounts to arbitrarily small perturbations in the moduli space of K\"ahler structures, and amounts to the smallest steps needed to resolve finite-distance singularities. Physically, one can interpret this as passing to the Higgs branch.\footnote{Actually, our procedure may modify the toric fan even if the original model was non-singular to begin with. In this case, of course, no harm is done to the original model.}

\subsection{Special case: Trilayer polytopes}\label{sec:trilayer_polytopes}
For definiteness, we now turn to a particularly simple class of Calabi-Yau orientifolds \cite{Kim:2020_Trilayer,Moritz:2023jdb}.
These arise by considering so-called \emph{trilayer polytopes}. 
Such polytopes are reflexive polytopes $\Delta^\circ \subset N_\R$ with additional structure: $\Delta^\circ$ must be equal to the convex hull of a point $v_*$ and a set of points $v_i$ lying in an affine hyperplane $H$ such that $-v_*\in H$ but $v_*\notin H$. 
The polytope $\Delta^\circ$ hence has a `pyramid' type shape where $v_*$ is the tip of the pyramid and the convex hull of the $v_i$ defines its base. 
An illustration is given in Fig.~\ref{fig:Trilayer_polytopes}.

It is not hard to see that the dual polytope $\Delta$ of a trilayer polytope is again trilayer and we will denote its special vertex by $\hat v_*$.
The points $\hat v_i$ are those in the dual facet of $v_*$.
Every point $v\in \Delta^\circ \cap N$ satisfies
\begin{align}
    -\langle \hat v_*,v\rangle \in \{-1,0,1\}\label{eq:trilayer_grading}
\end{align}
and this property gives the polytope its name.
Only the special point $v_*$ is mapped to $-1$ and all points in the hyperplane $H$ get mapped to $+1$ under the map above.
We will refer to these layers as the $1$--, $0$-- and $(-1)$--layers.

Given a trilayer polytope, it is always possible to find an isomorphism $N\to\Z^n$ such that the point $v_*$ gets mapped to $(-1,0,\ldots,0)$ and the points $v_i\in H$ get mapped to $(1,\star,\ldots,\star)$.

\begin{definition}
    The representation of $\Delta^\circ$ in which $v_*=(-1,0,...,0)$ and $v_i=(1,\star,\ldots,\star)$ is called the \emph{trilayer normal form} of $\Delta^\circ$.
\end{definition}

The three layers of $\Delta^\circ$ can now be easily identified as the spaces $\{z_1=-1\},\{z_1=0\},\{z_1=1\}$, with $z_1$ being the first coordinate of $\Z^n$.

\begin{figure}
        \centering
        \def\svgwidth{0.8\linewidth}
        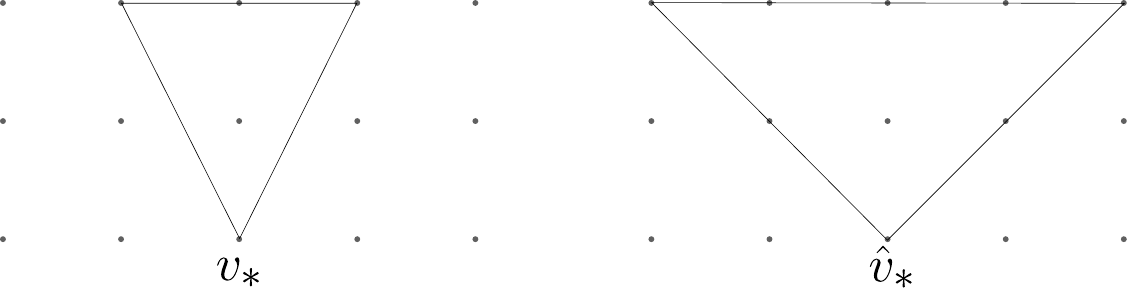
    \caption{The trilayer polytope of $\Proj[1,1,2]$ and its dual are displayed in their trilayer normal form. While the former polytope does not have any points in its zero-layer besides the origin, the latter has the two points $(0,1)$ and $(0,-1)$.}
    \label{fig:Trilayer_polytopes}
\end{figure}
Letting $V$ be a toric fourfold defined from a triangulation of $\Delta^\circ$, an involution of $V$ can be defined by the choice $\xi=v_*$. 

This involution has the following properties. 
First, there must exist at least one one-cone satisfying \eqref{eq:cones_fixed_locus}, namely $v_*$. 

Second, the projected-in monomials of the anticanonical divisor of $V$ --- those that can be used to define an O3/O7 orientifold --- are the monomials of $\Delta$ that are not in the zero-layer.
In case $\Delta^\circ$ does not contain any points in the zero layer besides the origin, the associated Calabi-Yau orientifold $B$ can itself be understood as a toric threefold. Indeed, $B$ can now be viewed as a hypersurface in $\widetilde{V}$ in the class of the prime toric divisor associated with the vertex $v_*$, and is hence isomorphic to a toric variety. Its toric fan is obtained from a new polytope $\widetilde\Delta$ generated by all points in the one-layer, with $-v_*$ representing the new origin (and triangulation induced from the underlying FRST of $\Delta^\circ$). 

Finally, there is a one-to-one correspondence between integral points in the $0$-layer, other than the origin, and NHCs.\footnote{This statement is due to Elijah Sheridan, and will appear in \cite{SMIIB}.} More precisely, for any such point $v$, we have that $v':=2v-v_*$ is in the $1$-layer  and satisfies $v'+v_*\in 2N$. Thus, $v'$ satisfies the criterion \eqref{eq:cones_fixed_locus} and the corresponding divisor hosts an O7-plane.

As an example, we consider weighted projective space $\Proj_{[1,1,2]}$, associated to the reflexive polytope displayed on the left-hand side of Fig.~\ref{fig:Trilayer_polytopes}. The polytope $\widetilde\Delta$ is the convex hull over the vertices $(1,-1)$ and $(1,1)$, which yield the toric fan of $\Proj^1$. 

\section{Systematic F--theory uplifts}\label{sec:Systematic_F_Theory_uplifts}
In this section we will construct the F--theory uplifts of the class of Calabi-Yau orientifold models obtained above.
For definiteness, we focus on the uplift of orientifold \emph{threefolds} to elliptically fibered Calabi-Yau fourfolds, but our methods straightforwardly generalize to any dimension.

We start with a Calabi-Yau orientifold $B_3=Y_3/\mathbb{Z}_2$, constructed as a hypersurface $\{P_B=0\}$ in a toric orbifold $\widetilde V_4=V_4/\mathbb{Z}_2$ as explained above. We will assume that $B_3$ is regular  (\emph{cf.} Def. \ref{def:regular_orientifold}). The goal is to construct a suitable elliptic fibration over $B_3$ in the form of a Weierstrass model,
\begin{equation}
    P_W=0\, ,\quad P_W:=y^2-x^3-fxz^4-gz^6\, ,
\end{equation}
with $f$ and $g$ suitable sections on $B_3$ and $[x:y:z]$ homogeneous coordinates on weighted projective space $\Proj_{[2,3,1]}$.

To this end, we will first assemble an appropriately twisted $\Proj_{[2,3,1]}$ fibration over the toric fourfold $\widetilde{V}_4$. This defines a toric sixfold that we will denote $V_6^s$. The complete intersection $\{P_B=P_W=0\}$ defines an elliptically fibered fourfold $Y_4^s$, which will turn out to be Calabi-Yau, as also discussed in \cite{Collinucci:2008zs,Collinucci:2009uh,Jefferson:2022ssj}.

This Calabi-Yau fourfold $Y_4^s$ will in general be singular due to degenerations of the elliptic fibration of type $I_0^*$ \cite{Morrison:2012np,Morrison:2012js,Morrison:2014lca,Halverson:2016vwx,Halverson:2017vde}. These singularities appear already in the Weierstrass hypersurface $\{P_W=0\}$ and are directly inherited by $Y_4^s$.

We will then systematically construct crepant resolutions of these singularities in $Y_4^s$ via a toric blow-up of the ambient sixfold $V_6\rightarrow V_6^s$. This corresponds to passing to the Coulomb branch of the corresponding M-theory compactification to three dimensions. The final result is an elliptically fibered Calabi-Yau fourfold $Y_4$, defined as a codimension-two complete intersection of hypersurfaces in $V_6$. It will be smooth away from point-like terminal $\mathbb{Z}_2$-orbifold singularities corresponding to orientifold three-planes.
In summary, we will move from left to right in the following illustration:
\begin{align}
    Y_3\subset V_4 \,\,\overset{\text{orbifold}}{\longrightarrow} \,\,B_3\subset \widetilde{V}_4\,\, \overset{\text{Weierstrass}}{\longrightarrow}\,\,Y_4^s\subset V_6^s\,\,\overset{\text{resolution}}{\longrightarrow}\,\, Y_4\subset V_6\,.\label{eq:Illustration}
\end{align}

In what follows, to avoid cluttered notation, we will not distinguish homogeneous coordinates, prime toric divisors and other associated quantities in notation between the toric varieties appearing in \eqref{eq:Illustration}.
We will likewise denote the restriction of an ambient divisor to a hypersurface in any of these varieties by the same symbol as the ambient divisor.
We hope that the variety the variables are associated with in each instance will be clear from the context. 
In general, will denote the homogeneous coordinates as $x_i$ and the associated divisors by $D_i$.

\subsection{The Weierstrass model}\label{sec:Weierstrass_model}
We begin with a complete, normal, and simplicial toric fourfold $\widetilde{V}_4$ defined via a toric orbifold $V_4/\mathbb{Z}_2$ as discussed in \S\ref{sec:toric_orbifolds}. As before, we let $B_3$ be a hypersurface  defined via the vanishing of a general global section $P_B$ of a sheaf $\mathcal{O}_{\widetilde{V}_4}(D_B)$, where $D_B\equiv \sum_i b_i D_i$ denotes the corresponding Weil-divisor. 

Using the adjunction theorem, we have that
\begin{equation}
    \overline{K}_{B}\sim \overline{K}_{\widetilde{V}_4}-D_B\, ,
\end{equation}
viewed as elements of the divisor class group.

In order to construct an F--theory uplift we now construct a Weierstrass model from a certain $\mathbb{P}_{[2,3,1]}$ fibration over $\widetilde{V}_4$:
\begin{equation}\label{eq:P231_fibration}
    \mathbb{P}_{[2,3,1]}\hookrightarrow V_6^s\overset{\pi}{\longrightarrow} \widetilde{V}_4\,,
\end{equation}
with $\pi$ denoting the projection onto the $4d$ toric base of the fibration.

The Weil divisor $D_B$ naturally pulls back to a divisor on $V_6^s$ which we will call, by the abuse of notation announced above, also $D_B$, and its associated polynomial $P_B$.
Explicitly, we construct the fibration $V_6^s$ as a simplicial toric sixfold in such a way that the simultaneous vanishing of $P_B$ and an appropriately chosen Weierstrass polynomial $P_W$ yields a Calabi-Yau fourfold as a codimension-two submanifold of $V_6^s$. 

We have
\begin{equation}\label{eq:Weierstrass_equation}
    P_W=y^2-x^3-fxz^4-gz^6\, ,
\end{equation}
with $(x,y,z)$ the homogeneous coordinates of the fiber $\mathbb{P}_{[2,3,1]}$, and $(f,g)$ sections of appropriately chosen sheaves $(\mathcal{O}_{\widetilde{V}_4}(D_f),\mathcal{O}_{\widetilde{V}_4}(D_g))$ on the base $\widetilde{V}_4$. We will fix these sheaves momentarily.

The fibration \eqref{eq:P231_fibration} is specified by letting the fiber coordinates $(x,y,z)$ take values in $(\mathcal{O}_{\widetilde{V}_4}(D_x),\mathcal{O}_{\widetilde{V}_4}(D_y),\mathcal{O}_{\widetilde{V}_4}(D_z))$ defined on $\widetilde{V}_4$. Without loss of generality we may choose $\mathcal{O}_{\widetilde{V}_4}(D_z)=\mathcal{O}_{\widetilde{V}_4}$, and then the Weierstrass polynomial \eqref{eq:Weierstrass_equation} makes sense if and only if the following conditions are met:
\begin{equation}
\mathcal{O}_{\widetilde{V}_4}(2D_y)=\mathcal{O}_{\widetilde{V}_4}(D_g)\, ,\quad \mathcal{O}_{\widetilde{V}_4}(2D_x)=\mathcal{O}_{\widetilde{V}_4}(D_f)\, ,\quad \mathcal{O}_{\widetilde{V}_4}(2D_g)=\mathcal{O}_{\widetilde{V}_4}(3D_f)\, .
\end{equation}

The hypersurface $\{P_W=0\}$ yields an elliptic fibration over $\widetilde{V}_4$. After imposing the further hypersurface constraint $P_B=0$, by using again the adjunction theorem, one obtains a Calabi-Yau fourfold if and only if
\begin{align}\label{eq:Weierstrass_line_bundles}
\mathcal{O}_{\widetilde{V}_4}\left({D_x}\right)&=\mathcal{O}_{\widetilde{V}_4}\left({2\overline{K}_{B}}\right)\, ,\quad \mathcal{O}_{\widetilde{V}_4}\left({D_y}\right)=\mathcal{O}_{\widetilde{V}_4}\left({3\overline{K}_{B}}\right)\, ,\\ 
\mathcal{O}_{\widetilde{V}_4}\left({D_f}\right)&=\mathcal{O}_{\widetilde{V}_4}\left({4\overline{K}_{B}}\right)\, ,\quad \mathcal{O}_{\widetilde{V}_4}\left({D_g}\right)=\mathcal{O}_{\widetilde{V}_4}\left({6\overline{K}_{B}}\right)\, .
\end{align}
Hence, the fibration \eqref{eq:P231_fibration} is completely determined by requiring that the Weierstrass model \eqref{eq:Weierstrass_equation} makes sense and that the codimension-two hypersurface $\{P_W=P_B=0\}$ is Calabi-Yau.

It is now straightforward to construct a toric fan $\Sigma_6^s$ for $V_6^s$. We denote by $\widetilde{\Sigma}_4(1)$ the one-dimensional cones of the toric fan $\widetilde{\Sigma}_4$ of $\widetilde{V}_4$, and by the set $\{{v}_1,\ldots,{v}_n\}\subset \widetilde{N}_4\simeq \mathbb{Z}^4$ the primitive generators of the $1d$ cones. We embed the $\{v_i\}$ into a six-dimensional lattice $N_6\simeq \mathbb{Z}^6$ via
\begin{equation}\label{eq:4d_to_6d_embedding}
    v_i\mapsto \begin{pmatrix}
        v_i\\
        w_i
    \end{pmatrix}\, ,\quad w_i\in \mathbb{Z}^2\, ,
\end{equation}
and define primitive generators,
\begin{equation}
    v_x=\begin{pmatrix}
        0\\
        3\\
        1
    \end{pmatrix}\, ,\quad v_y=\begin{pmatrix}
        0\\
        -2\\
        -1
    \end{pmatrix}\, ,\quad v_z=\begin{pmatrix}
        0\\
        0\\
        1
    \end{pmatrix}\, ,
\end{equation}
associated with the homogeneous fiber coordinates $(x,y,z)$. 

This assignment \eqref{eq:4d_to_6d_embedding} satisfies the constraints \eqref{eq:Weierstrass_line_bundles} if and only if
\begin{equation}\label{eq:linear_equation_determining_wi}
    \sum_i \eta_i \left(w_i+(b_i-1)\begin{pmatrix}
        0\\
        1
    \end{pmatrix}\right)\in \mathbb{Z}^2\quad \forall \,\eta_i \,\,\text{s.t.:}\quad \sum_i \eta_i v_i\in \mathbb{Z}^4 \, ,
\end{equation}
where, as defined before, $D_B\equiv \sum_i b_i D_i$. Any choice of 2d lattice points $w_i$ that solves \eqref{eq:linear_equation_determining_wi} leads to an equivalent $\mathbb{P}_{[2,3,1]}$ fibration, so without loss of generality we may set
\begin{equation}
    w_i^{(1)}=0\, ,\quad w_i^{(2)}=1-b_i\, .
\end{equation}

This fixes the rays of the fan $\Sigma_6^s$.
The appropriate higher dimensional cones of $\Sigma_6^s$ are obtained as follows: for each top-dimensional cone $\sigma^k$ of $\tilde\Sigma$, we define three $6$-dimensional cones in $\Sigma_6^s$ by adding the combinations $(v_x,v_y),(v_y,v_z),(v_z,v_x)$ to the primitive generators of $\sigma^k$.
These cones, together with their faces, define the toric fan $\Sigma_6^s$.

For completeness, we note that the GLSM charge matrix $Q^6_{ij}$ of $V_6^s$ can be represented as an extension of the GLSM charge matrix $Q^4_{ij}$ of $\widetilde V_4$ by one row and three columns. 
The columns contain the charges of the three additional rays $v_x,v_y,v_z$ and the additional row simply contains the relation $2v_x+3v_y+v_z=0$.
The entries of the three columns can be determined by using the relation \eqref{eq:Weierstrass_line_bundles}.
The entries $Q_{ix}^6,Q_{iy}^6,Q_{iz}^6$ of the three extended columns are given by
\begin{align}
    Q_{ix}^6 = 2\cdot \sum_jQ_{ij}^6(1-b_j)\,,\qquad Q_{iy}^6 = 3\cdot\sum_j Q_{ij}^6(1-b_j)\,,\qquad Q_{iz}^6 = 0\, .
\end{align}

The corresponding element in the divisor class group associated to the Weierstrass hypersurface can be represented by
\begin{align}
    D_W = 3D_x =2D_y\,,\label{eq:Weierstrass_Weil_divisor}
\end{align}
as is clear from the form of the Weierstrass equation \eqref{eq:Weierstrass_equation}.
\subsection{Resolution of codimension-one singularities}\label{sec:smoothening}
The generic member of the family of Weierstrass hypersurfaces $\{P_W=0\}$ in $V_6^s$ can be singular, even if $\widetilde{V}_4$ is smooth. 
For any given member of the family, singularities arise when the elliptic fiber degenerates over loci in the base, and the theory of the occurring singularities is well studied in the mathematics and the F--theory literature \cite{Kodaira1963OnCA,PMIHES_1964__21__5_0,Vafa:1996xn,Weigand:2018rez}.

The generic member of the Weierstrass family is singular, if the singularity is present for all members of the family. 
According to Bertini's Theorem \ref{thm:Bertini}, such generic singularities are located in the basepoint locus of $D_W$.
From the form of the Weierstrass equation \eqref{eq:Weierstrass_equation}, we can deduce that the basepoint locus takes the form
\begin{align}
    BL_{D_W} = \{x=y=0\}\cap \pi^*(BL_{6\overline{K}_B})\,,
\end{align}
where $BL_{6\overline{K}_B}$ is the basepoint locus of $6\overline{K}_B$ in $\widetilde{V}_4$.
The set $\{x=y=z=0\}$ is not part of $V_6^s$ because the three corresponding rays do not form a cone together.
The basepoint loci at codimension-two and higher depend on the triangulation of the fan, whereas the basepoints at codimension-one do not. 
We hence focus on these basepoints first and refer to \S\ref{sec:Cartier_and_nef} for a  treatment of the higher-codimension basepoints.

A codimension-one basepoint locus of $6\overline{K}_B$ along $D_i$ arises when every section of $6\overline{K}_B$ vanishes on the locus $\{x_i=0\}$.
This signals a degeneration of the elliptic fiber along $D_i$, or, in the type IIB language, the presence of seven-branes on this divisor. 
Due to \textit{every} section vanishing along $D_i$, the seven-branes are rigid and the corresponding gauge group cannot be Higgsed by perturbing the seven-brane embedding. 
Hence, such loci are called \textit{non-Higgsable clusters} (NHCs) \cite{Morrison:2012np,Morrison:2012js,Morrison:2014lca,Halverson:2016vwx,Halverson:2017vde}.

It is well known \cite{Sen:1996vd,Sen:1997gv,Halverson:2016vwx} that in F--theory models corresponding to type--IIB orientifold models, the fiber degeneration along every NHC has to be of $I_0^*$--type and the corresponding gauge algebra is thus $\mathfrak{so}(8),\mathfrak{so}(7)$ or $\mathfrak{g}_2$.
Correspondingly, the polynomials $f,g$ and the discriminant $\text{Disc}(P_W)$ factorize as follows
\begin{align}\label{eq:NHC_f_g}
    &f \equiv  \prod_{i\in \text{NHCs}}{x}_i^{2}f'
    ,\quad g \equiv \prod_{i\in \text{NHCs}}{x}_i^{3}g'
    \,, \\ 
    &\text{Disc}(P_W):=4f^3+27g^2\equiv \prod_{i\in \text{NHCs}} x_i^6\text{Disc}'(P_W)\, ,
\end{align}
with $f',g',\text{Disc}'(P_W)$ generically non-vanishing along all the NHC divisors $\{x_i=0\}$.

The generic member of the family $\{P_W=0\}$ is then singular along any of the loci $\{x=y= x_i=0\}$, for $D_i$ hosting an NHC, because the Jacobian criterion is violated, i.e., $P_W$ and $dP_W$ vanish simultaneously.

These singularities cannot be resolved by complex structure deformations, but a K\"ahler resolution is available, and corresponds to moving onto the Coulomb branch of the gauge theory (in three dimensions).
This branch is geometrically realized as a blow-up morphism, introducing exceptional divisors $E_1,\ldots,E_{{\rm rank} (G)}$ in the (five-dimensional) Weierstrass hypersurface $W=\{P_W=0\}$ that descend to divisors in $\widetilde{V}_4$. 

This blow-up also introduces exceptional curves $\tilde{E}_1,\ldots,\tilde{E}_{{\rm rank}(G)}$ --- all of which are rational curves --- and turns the degenerate elliptic fibers into chains of $\Proj^1$s.

The exceptional curves intersect the exceptional divisors as
\begin{align}
    \left(E_a\cdot \tilde{E}_b\right)_{W} = -C_{ab}\,,\label{eq:intersection_D4}
\end{align}
where the subscript $W$ denotes that the intersection is taken in the Weierstrass hypersurface and $C_{ab}$ is the Cartan matrix of the gauge algebra. For the three gauge algebras of consideration, we have
\begin{align}
     {C}_{ab}^{ \mathfrak{so}(8)} = \begin{pmatrix}
        2 & -1 & -1 & -1\\
        -1 & 2 & 0 & 0\\
        -1 & 0 & 2 & 0\\
        -1 & 0 & 0 & 2
    \end{pmatrix}\,,\; {C}_{ab}^{ \mathfrak{so}(7)} = \begin{pmatrix}
        2 & -1 & -2\\
        -1 & 2 & 0\\
        -1 & 0 & 2
    \end{pmatrix}\,,\; {C}_{ab}^{\mathfrak{g}_2} = \begin{pmatrix}
        2 & -3\\
        -1 & 2
    \end{pmatrix}\,.\label{eq:Cartan_matrix}
\end{align}

The $I_0^*$ fiber leads to a singularity at codimension three, and, in the language of toric geometry,  corresponds to the three-cone with primitive generators $v_i,v_x,v_y$.
One might na\"ively expect that the normal directions to this codimension-three locus already yield a curve worth of $D_4$ singularities, but this is not possible: $A_n$-type singularities are the only ones realized as toric surfaces. 

In fact, the cone generated by $(v_i,v_x,v_y)$ is generally smooth and the singularity only manifests itself on the Weierstrass hypersurface.
Hence, in order to blow-up the singularity, we have to perform blow-ups of \textit{regular} points of $V_6^s$.
One might expect that the number of necessary blow-ups in the cone is defined by the rank of the present gauge algebra but, as we will explain below, it is sufficient to introduce two exceptional divisors in the toric sixfold for all three cases.
For the case of $\mathfrak{so}(8)$ ($\mathfrak{so}(7)$) one of the exceptional divisors in the toric sixfold descends to three (two) divisors on the Weierstrass hypersurface, while in the case of $\mathfrak{g}_2$ it remains irreducible.

To be fully explicit, let us consider the Weierstrass equation with one NHC on $D_i$ with $f,g$ as in \eqref{eq:NHC_f_g}.
We now refine $\Sigma_6^s$ to $\Sigma_6$, the toric fan of the blown-up variety $V_6$, by the introduction of two additional rays, with primitive generators
\begin{align}\label{eq:blowups_Coulomb_branch}
    v_{e_1} = {v}_i+v_x+2v_y\,,\qquad v_{e_2} = 2 v_i+2 v_x + 3 v_y\, ,
\end{align}
the choice of which will be justified \emph{a posteriori}. We will refer to the homogeneous coordinates associated with these rays as $e_1$ and $e_2$.

The three-cone generated by $(v_i,v_x,v_y)$ gets refined by the introduction of the blow-up rays, and admits a unique regular triangulation. As a consequence, every simplicial cone in $\Sigma_6^s$ containing ${v}_i,v_x,v_y$ as generators is replaced by a set of five simplicial cones, according to the replacement rule
\begin{align}
    \{( v_i,v_x,v_y)\} \to \{( v_i,v_x,v_{e_2}),( v_i,v_y,v_{e_1}),( v_i,v_{e_1},v_{e_2}),(v_x,v_y,v_{e_1}),(v_x,v_{e_1},v_{e_2})\}\,.\label{eq:cones_resolution}
\end{align}
The local GLSM of the relevant local three dimensional toric variety then takes the form 
\begin{align}\label{eq:local_GLSM}
    \begin{pmatrix}
        {x}_i & x & y & e_1 & e_2\\
        1 & 1 & 2 & -1 & 0\\
        2 & 2 & 3 & 0 & -1
    \end{pmatrix}\, .
\end{align}
From this, one confirms that after passing to the blown-up variety, the Weierstrass hypersurface equation can be written as
\begin{align}
    \begin{split}
        0 = P_W =e_1y^2-x^3-x_i^2 f'xz^4-x_i^3 g'z^6\,.\label{eq:Weierstrass_equation_resolved}
    \end{split}
\end{align}
The polynomial $P_W$ still identically vanishes along the locus $\{y=x= x_i=0\}$ and now also along $\{e_1=x=x_i=0\}$, but, as can be seen from the cone assignments \eqref{eq:cones_resolution}, neither combination of associated rays forms a cone together, such that the dangerous-looking loci are actually not part of $V_6$.
This means that the basepoint locus associated to the NHC at ${x}_i=0$ is resolved.

After setting $e_1=0$ in \eqref{eq:Weierstrass_equation_resolved}, $P_W$ takes the form
\begin{align}\label{eq:mondromy_polynomial}
    -P_W|_{e_1=0} &= x^3+x_i^2f'xz^4+x_i^3g'z^6\,.
\end{align}
This is precisely the ``monodromy polynomial'' whose splitting properties govern the surviving gauge algebra \cite{tate1975algorithm,Bershadsky:1996nh,Katz:2011qp,Grassi:2011hq,Weigand:2018rez}. 
Indeed, the monodromy polynomial may factorize into two, or three components, or be irreducible. 

More precisely we have
\begin{align}
    P_W|_{e_1=0}^{\mathfrak{so}(8)} &= (x-x_iQz^2)(x-x_iRz^2)(x+x_i(Q+R)z^2)\,,\label{eq:Weierstrass_three_splitting}\\
    P_W|_{e_1=0}^{\mathfrak{so}(7)} &= (x-x_iQz^2)(x^2+x_iQ x z^2+x_i^2Rz^4)\,,\label{eq:Weierstrass_two_splitting}\\
    P_W|_{e_1=0}^{\mathfrak{g}_2} &= x^3+x_i^2Qxz^4+x_i^3Rz^6\,,\label{eq:Weierstrass_one_splitting}
\end{align}
where $Q,R$ are sections of appropriate line bundles on the base of the fibration.
The divisor $\{e_1=0\}$ hence descends to $(3,2,1)$ divisors on the Weierstrass hypersurface for the gauge algebras $(\mathfrak{so}(8),\mathfrak{so}(7),\mathfrak{g}_2)$ \cite{tate1975algorithm,Bershadsky:1996nh,Katz:2011qp,Grassi:2011hq,Weigand:2018rez}.

We will now justify the choice of blow-up morphism \eqref{eq:blowups_Coulomb_branch}. We will do so by verifying that the intersection pairing of the exceptional components indeed leads to \eqref{eq:intersection_D4}.\footnote{We refer to App.~\ref{app:intersections_non_favorable} for details about computing intersection numbers in the relevant ``non-favorable'' situations, in which the hypersurface does not inherit a basis of its second cohomology group from the variety it is embedded in.} 

Locally, in a tubular neighborhood around the locus hosting the NHC, the total space $V_6$ can be understood as a fibration of the three-dimensional toric variety $V_3$, defined by the cones \eqref{eq:cones_resolution}, over a three complex-dimensional base that hosts the $I_0^*$ singularity.
The total Weierstrass hypersurface is then the $2d$ (resolved) Weierstrass surface $W_2$ in $V_3$ fibered over the same base space.
In the local toric threefold, the Weierstrass polynomial trivially splits
\begin{align}\label{eq:split_local_threefold}
    P_{W}|_{e_1=0}=(x-\lambda_1)(x-\lambda_2)(x-\lambda_3)\,,
\end{align}
as all coefficients of the remaining polynomial are simply complex numbers. 
Hence, there will always be four exceptional curves/divisors $E_i,\tilde{E}_i$, $i=1,2,3,4$ in $W_2$.
Using basic intersection theory one arrives at the intersection matrix
\begin{align}
    D_W\cdot D_{e_a}\cdot D_{e_b}|_{V_3} = \begin{pmatrix}
        -6 & 3\\
        3 & -2
    \end{pmatrix}\,,\label{eq:e1_e2_threefold}
\end{align}
where we have indicated that the intersection is taken in the local toric threefold.

As the splitting of $\{e_1=0\}$ into three components $E_i$, $i=1,2,3$ \eqref{eq:split_local_threefold} is completely democratic, we can deduce that the intersection matrix on the $2d$ Weierstrass hypersurface $W_2$ within $V_3$ is given by
\begin{align}\label{eq:intersection_form_loc_threefold}
    E_a\cdot \tilde{E}_b|_{W_2} = \begin{pmatrix}
        -2 & 1 & 1 & 1\\
        1 & -2 & 0 & 0\\
        1 & 0 & -2 & 0\\
        1 & 0 & 0 & -2
    \end{pmatrix}\,,\qquad a,b\in \{1,\ldots,4\}\,.
\end{align}

But, divisors of $W_2$ need not define divisors of the $5d$ Weierstrass hypersurface upon fibering over the $3d$ base.
It may be that only suitable unions of divisors in $W_2$ define divisors in $W$. Whenever this happens, a subset of the exceptional curves become homologous.
The monodromy responsible for the folding of the corresponding Dynkin diagram from $D_4$ to $B_3$ or $G_2$ precisely explains how the homology classes of $W_2$ relate to those of $W$:
\begin{enumerate}
    \item For the $\mathfrak{so}(8)$ case, there is no monodromy and each divisor/curve in $W_2$ also defines a divisor/curve in $W$. Hence, Eq.~\eqref{eq:intersection_form_loc_threefold} defines the intersection form between exceptional divisors and curves in $W$ and we match the required intersection form \eqref{eq:Cartan_matrix}.
    \item For the $\mathfrak{so}(7)$ case, there is a monodromy that relates two of the exceptional $\Proj^1$s in $W_2$. As a result, only the union of two divisors in $W_2$ actually defines a divisor in $W$.
    \item For the $\mathfrak{g}_2$ case, only the sum of all three components of $e_1$ defines a divisor in $D_W$ and all three associated curve classes are in the same homology class in $W$. 
\end{enumerate}
The corresponding intersection forms are thus recovered by adding two (three) columns of \eqref{eq:intersection_form_loc_threefold} to each other and then only keeping the independent rows:
\begin{align}
    E_a\cdot \tilde{E}_b|_W^{\mathfrak{so}(7)}=\begin{pmatrix}
        -2 & 1 & 2\\
        1 & -2 & 0\\
        1 & 0 & -2
    \end{pmatrix}\,,\qquad E_a\cdot \tilde{E}_b|_W^{\mathfrak{g}_2}=\begin{pmatrix}
        -2 & 3\\
        1 & -2
    \end{pmatrix}\,.
\end{align}
We thus recover the intersection forms \eqref{eq:intersection_D4},\eqref{eq:Cartan_matrix} and we have therefore shown that the toric blow-ups \eqref{eq:blowups_Coulomb_branch} define the K\"ahler resolution of the NHC singularities as claimed. 

A couple of comments are in order: first, for the $\mathfrak{so}(8)$ and $\mathfrak{so}(7)$ cases, the divisors of $Y_4$ are not directly inherited from $H^2(V_6,\mathbb{Z})$, and the embedding of $Y_4$ into $V_6$ is thus ``not favorable''.

Second, we note that the weights of the divisor $D_B$ in the blown-up variety can straightforwardly be determined from its weights before the blow-up, using the form of the local GLSM \eqref{eq:local_GLSM}.

Finally, we observe that in our main case of interest, uplifts of globally defined Calabi-Yau orientifolds (without intersecting O7 planes), the gauge algebra of all NHCs is $\mathfrak{so}(8)$. 
This can be seen as follows: $2\overline{K}_B$ is the class of the total O7-plane. Factoring out a single factor $x_i$ and restricting to $x_i=0$ the general section vanishes precisely along all the curves along which other orientifold seven-planes intersect our reference O7-plane $D_i$. But as pairs of O7-planes do not intersect, by assumption, this general section must actually be constant. An analogous statement also holds for $4\overline{K}_B$ and $6\overline{K}_B$, after factoring out $x_i^2$ and $x_i^3$ respectively and restricting to $x_i=0$, so the sections $f'|_{x_i=0}$ and $g'|_{x_i=0}$ are also constant functions. We conclude that the monodromy polynomial is simply a cubic in $\mathbb{P}^1$, and therefore factorizes globally along the locus $x_i=0$, leaving no room for a monodromy that could break the gauge algebra to $\mathfrak{so}(7)$ or $\mathfrak{g}_2$.

In \S\ref{sec:5d_uplift}, we will briefly discuss situations where pairs of O7 planes do intersect in certain singular collision limits (as in \cite{DelZotto:2014hpa}), and then recombine to a single seven-brane stack with gauge algebra reduced to $\mathfrak{g}_2$ or $\mathfrak{so}(7)$. We believe such models still have a well-behaved weakly coupled type IIB description in terms of D7-branes on O7-planes, even though they do not directly arise from a globally defined orientifold involution of a Calabi-Yau threefold.

\section{Properties of the uplift}\label{sec:Properties}
We now turn to collecting  properties of the toric variety $V_6$ and the Calabi-Yau submanifold $Y_4$ therein.
We will mostly be interested in identifying conditions on the divisors $D_B$ and $D_W$ such that the mirror of $Y_4$ can be understood in a combinatorial way, according to the works of Batyrev and Borisov \cite{borisov1993mirrorsymmetrycalabiyaucomplete,batyrev1994dualconesmirrorsymmetry}. 
We focus on these, because we are ultimately interested in the superpotential of F--theory compactifications, which is determined by the periods of $Y_4$.
In the large complex structure limit of $Y_4$, the periods can be systematically computed if the corresponding mirror manifold $X_4$ is known explicitly as a complete intersection of hypersurfaces in a toric variety \cite{Hosono:1993qy,Hosono:1994ax}.
\subsection{Nef-partitions}\label{sec:nef_partitions}
As discussed in \S\ref{sec:toric_varieties}, mirror symmetry can be understood in a combinatorial way for Calabi-Yau hypersurfaces if the toric fan of the ambient variety $V$ derives from a triangulation of a reflexive polytope \cite{Batyrev:1993oya}. 
The polytope being reflexive is equivalent to the anticanonical divisor $\overline{K}_V$ of $V$ being Cartier and nef.

Batyrev's construction generalizes to complete intersection Calabi-Yaus (CICYs) that are cut out by $k$ Cartier and nef divisors $D^{(1)},\ldots,D^{(k)}$ in a toric variety, if the  $D^{(1)},\ldots,D^{(k)}$ form a \emph{nef-partition} \cite{borisov1993mirrorsymmetrycalabiyaucomplete,batyrev1994dualconesmirrorsymmetry}, a property we shall define below.

In App.~\ref{app:nef_partitions}, we review the Batyrev-Borisov mirror symmetry construction of \cite{borisov1993mirrorsymmetrycalabiyaucomplete,batyrev1994dualconesmirrorsymmetry} for general $k$-part CICYs. Here, we focus on the special case $k=2$ applied to the F--theory uplifts introduced in \S\ref{sec:Systematic_F_Theory_uplifts}.

As before, we let $D_B,D_W$ have decompositions
\begin{align}
    D_B=\sum_i b_iD_i\,,\qquad D_W=\sum_iw_iD_i\, ,\label{eq:DB_DW_decomposition}
\end{align}
in terms of the prime toric divisors $D_i$.
As by construction $D_B+D_W=\overline{K}_{V_6}$, without loss of generality, we can set $b_i+w_i=1$.

If $D_B,D_W$ are Cartier and nef, the Minkowski sum of their Newton polytopes
\begin{align}\label{eq:Delta_polytope}
    \Delta:=\Delta_B+\Delta_W\subset M_\R\,,
\end{align}
is then reflexive, \emph{cf.} Lemma \ref{lemma:normal_fan}.
The polar dual of \eqref{eq:Delta_polytope}, denoted $\Delta^\circ$, is then equal to the convex hull over the one-skeleton
\begin{align}
    \Delta^\circ = \text{Conv}(\Sigma_6(1))\subset N_\R\,.
\end{align}

\begin{definition}
    Let $D_B,D_W$ be Cartier and nef divisors on a toric variety $V$ such that $D_B+D_W=\overline{K}_V$.
    Such a decomposition of the anticanonical divisor is called a \emph{nef-decomposition}.
\end{definition}

\begin{definition}[\cite{borisov1993mirrorsymmetrycalabiyaucomplete,batyrev1994dualconesmirrorsymmetry}]
    Let $D_B,D_W$ as in \eqref{eq:DB_DW_decomposition} be a nef-decomposition. If in addition, 
    \begin{align}
        b_i,w_i\in \{0,1\}\,,\label{eq:partition_condition_2}
    \end{align}
    we say that $D_B,D_W$ form a \textit{nef-partition}.
\end{definition}

We note that $D_B,D_W$ forming a nef-partition implies that
    $\text{Conv}(\Delta_B,\Delta_W)$
is a reflexive polytope as well, see App.~\ref{app:nef_partitions}.

The property \eqref{eq:partition_condition_2} naturally partitions the rays of $\Sigma_6$ into two sets
\begin{align}
    E^B=\{v_i\in \Sigma_6(1)\;|\;b_i=1\}\,,\qquad E^W=\{v_i\in \Sigma_6(1)\;|\;w_i=1\}\,.
\end{align}
From these we define the polytopes
\begin{align}
    \nabla_B:=\text{Conv}(\{0\}\cup E^B)\,,\qquad \nabla_W:=\text{Conv}(\{0\}\cup E^W)\, .
\end{align}
Their Minkowski sum forms another reflexive polytope
\begin{align}
    \nabla := \nabla_B+\nabla_W\, ,
\end{align}
and its reflexive dual satisfies
\begin{equation}
    \nabla^\circ \equiv \text{Conv}(\Delta_B,\Delta_W)\, .
\end{equation}

An FRST of $\nabla^\circ$ now defines a toric variety $V_6^\vee$ and the polytopes $\nabla_B,\nabla_W$ define Cartier and nef divisors $D_B^\vee,D_W^\vee$ on $V_6^\vee$.
The corresponding CICY
\begin{align}
    X_4=D_B^\vee\cap D_W^\vee \subset V_6^\vee\,,
\end{align}
is the mirror dual of $Y_4$ \cite{borisov1993mirrorsymmetrycalabiyaucomplete,batyrev1994dualconesmirrorsymmetry}.
An illustration of all relevant polytopes is given in Fig.~\ref{fig:BB_Mirror_symetry}.

\begin{figure}
        \centering
        \includegraphics[width=0.8\linewidth]{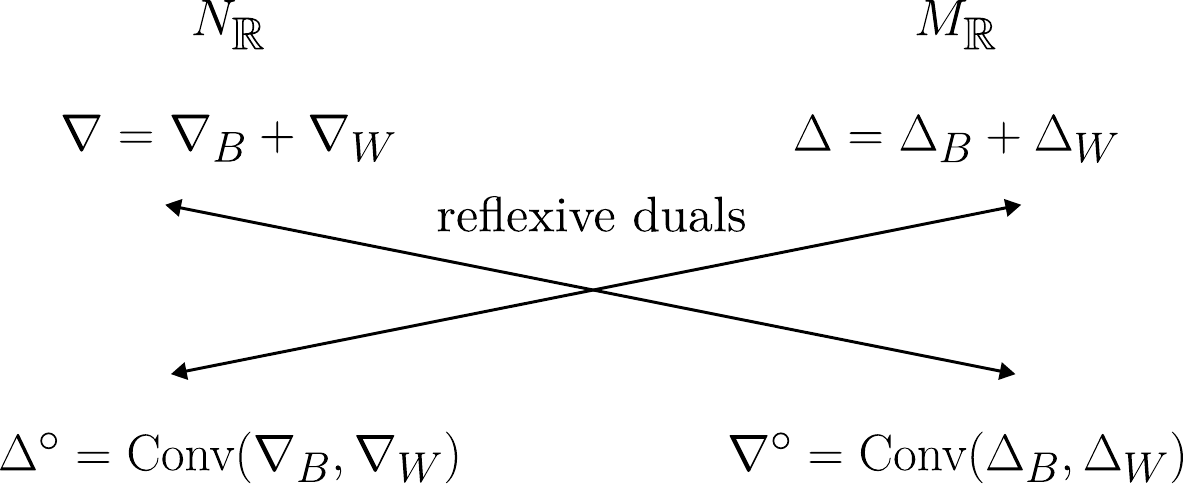}
        \caption{Dualities for Batyrev-Borisov mirror-symmetry \cite{borisov1993mirrorsymmetrycalabiyaucomplete,batyrev1994dualconesmirrorsymmetry} coming from nef-partitions.}
        \label{fig:BB_Mirror_symetry}
\end{figure}

\subsection{Hodge numbers}
If $Y_4$ is given by a nef-partition, we can compute the Hodge numbers combinatorially, as reviewed in App.~\ref{app:Hodge_numbers}.
The upshot is that there exists a generating functional for the Hodge numbers of CICYs realized by $k$-part nef-partitions \cite{Batyrev:1994pg,Batyrev:1995ca}.
However, the extraction of $h^{1,1},h^{2,1}$ and $h^{3,1}$ becomes increasingly more complicated with growing $k$.
For our main case of interest, $k=2$, the explicit formulas are given by \eqref{eq:Hodge_numbers},\eqref{eq:Hodge_number_h21}.

As is well known, $h^{1,1}(Y_4)$ is related to $h^{1,1}(B_3)$ by the Shioda-Tate-Wazir Theorem \cite{Shioda1972,wazir2001arithmeticellipticthreefolds} via
\begin{align}
    h^{1,1}(Y_4)=h^{1,1}(B_3)+1+4\cdot n_{\rm NHC}+n_{U(1)}\,,
\end{align}
with $n_{\rm NHC}$ the number of NHCs and $n_{U(1)}$ the number of additional sections of the elliptic fiber which correspond to global $U(1)$ fields in the type IIB description.\footnote{There exist examples of F-theory models with `non-Higgsable $U(1)$s' \cite{Wang:2016urs}. In contrast, we do not expect non-Higgsable $U(1)$s to arise in our setting where $Y_4$ is the F--theory uplift of a weakly coupled type IIB orientifold model with its base $B_3$ constructed as a hypersurface in a toric variety. It would be very interesting to determine whether such models exist.}

\subsection{Triangulations of \texorpdfstring{$\nabla$}{nabla}}
In \S\ref{sec:Systematic_F_Theory_uplifts}, we constructed a toric fan $\Sigma_6$ of $V_6$, the ambient space of the uplift $Y_4$. Let us assume that $D_B,D_W$ form a nef-decomposition and hence the polytope $\Delta^\circ = \text{Conv}(\{v_i\,, v_i\in \Sigma_6(1)\})$ is reflexive. 

Then, this toric  fan arises from a regular star triangulation of the polytope $\Delta^\circ = \text{Conv}(\{v_i\,, v_i\in \Sigma_6(1)\})$.\footnote{This is because by construction, the anticanonical divisor is Cartier and nef.} This triangulation is however not in general fine, because it need not include all the points associated with divisors that do not intersect the Calabi-Yau fourfold.

One may, however, always refine such fans to form \emph{fine} regular star triangulations (FRSTs) of $\Delta^\circ$ (subject to compatibility with the $\Proj^2_{[2,3,1]}$ fibration \eqref{eq:P231_fibration}). Let $\Sigma_6'$ be one such fan.
By construction, the subvariety $D_B\cap D_W$ is the same as our original Calabi-Yau fourfold $Y_4$.

The fan $\Sigma_6'$ can be finer than $\Sigma_6$ because we are taking a convex hull. 
Divisors associated to points interior to facets or codimension-two faces of $\Delta^\circ$ will not intersect $Y_4$, for the same reasons that divisors associated to points interior to facets of a polytope do not intersect the associated Calabi-Yau hypersurface \cite{Batyrev:1993oya}.
Hence, these rays can show up in $\Sigma_6'$ while they in general do not in $\Sigma_6$.
In the presence of NHCs, additional such rays in $\Sigma_6'$ do exist: denoting the ray of an NHC by $v_i$ and its two associated blow-ups by $v_{e_1},v_{e_2}$, the convex hull of $\{0,v_i,v_x,v_y,v_{e_1},v_{e_2}\}$ contains one additional point given by $v_i+v_x+v_y$.
More points can (and in general will) appear but all of these have the property that they do not intersect the Calabi-Yau complete intersection.

Next, both the Cartier and nef properties of divisors depend on the triangulations of the fan. It then na\"ively appears that additional constraints on the set of FRSTs of $\Delta^\circ$ are needed in order to ensure that the divisors $D_B$ and $D_W$ are Cartier and nef. After all, only the divisor $D_B+D_W=\overline{K}_{V_6}$ is Cartier and nef by construction. But, in fact any FRST leads to $D_B$ and $D_W$ being Cartier and nef, since every FRST is a refinement of the normal fan of $\overline{K}_{V_6}$. 
By Theorem \ref{thm:normal_fan_of_Minkowski_sum}, $\Sigma_{\Delta_M^B+\Delta_M^W}$ is a refinement of both normal fans $\Sigma_{\Delta_B}$ and $\Sigma_{\Delta_{W}}$. 

\subsection{Building a nef-partition}\label{sec:Cartier_and_nef}
According to Theorem \ref{thm:Bertini}, $Y_4$ is guaranteed to be smooth away from singularities of $V_6$ if $D_B$ and $D_W$ are basepoint free divisors. Thus, to ensure smoothness --- and for applying the mirror symmetry dictionary reviewed in \S\ref{sec:nef_partitions} --- we are interested in finding models in which $D_B,D_W$ are Cartier and nef \emph{by construction}.

In \S\ref{sec:Systematic_F_Theory_uplifts}, we have explained how codimension-one basepoint loci of the Weierstrass line bundle $\mathcal{O}(D_W)$ are resolved. 
These did not depend on triangulation data and required a blow-up morphism for resolution.
Here, we turn to higher-codimension basepoint loci.

Let us start by recalling that after resolution of $I_0^*$ singularities the Weierstrass equation takes on the form
\begin{align}
    P_W=y^2\prod_{i\in \text{NHC}} e_{i,1} - x^3 - f'(x)xz^4\prod_{i\in \text{NHC}}x_i^2 - g'(x)z^6 \prod_{i\in \text{NHC}} x_i^3\,.
\end{align}
From this we see that while the basepoint loci at $\{x=y=x_i\}=0$ have been resolved by introducing the exceptional coordinates $e_{i,1},e_{i,2}$, there now may arise new basepoint loci at $\{e_{i,1}=x=x_j=0\}$ for $i\neq j$.

But, while $P_W$ would vanish along such loci, it is not guaranteed that these loci actually are part of $V_6$. 
This of course depends on whether the pairs $v_i,v_j$ form cones together in the fan of $\widetilde{V}_4$, or, in other words, whether any pairs of rigid O7 planes intersect.

If a pair of rigid O7-planes intersect, the codimension-four subset $\{e_{i,1}=x=y=x_j=0\}$ in $V_6$ is singular as both $P_W$ and $dP_W$ vanish on this subspace. Such ``collision loci'' were studied in \cite{DelZotto:2014hpa} and arise at infinite distance in moduli space. Nevertheless, they can be resolved by simply blowing up the collision locus. Finally, basepoint loci of $D_W$ may arise as higher-codimension basepoint loci of $6\overline{K}_B$ in $\widetilde{V}_4$.

A sufficient, and plausibly necessary, condition for $D_W$ and $D_B$ to be Cartier and nef is encapsulated in the following\footnote{We are not aware of examples where the criteria (a)-(c) of Lemma \ref{lem:uplift_is_refl} do not hold and yet the divisors $D_{B}$, $D_W$ nevertheless end up being Cartier and nef.}

\begin{lemma}\label{lem:uplift_is_refl}
Let $B_3=D_B\subset \widetilde V_4$ be a Calabi-Yau orientifold hypersurface, and let $Y_4=D_B\cap D_W\subset V_6$ be its F--theory uplift after resolving all codimension-one $I_0^*$ singularities as in \S\ref{sec:smoothening}. Let
\begin{equation}\label{eq:KB6_movable}
    D_{g'}:=6\overline K_B-\sum_{i\in{\rm NHC}}3D_i
\end{equation}
denote the movable part of $6\overline K_B$, obtained by subtracting the fixed components associated with the non-Higgsable $I_0^*$ divisors. Suppose that:
\begin{itemize}
    \item[(a)] $D_B$ is Cartier and nef on $\widetilde V_4$.
    \item[(b)] $D_{g'}$ is Cartier and nef on $\widetilde V_4$.
    \item[(c)] Rigid O7 divisors do not intersect.
\end{itemize}
Then the divisors $D_B$ and $D_W$ which define the resolved F--theory uplift are both Cartier and nef on $V_6$.

\noindent\textbf{Proof:} Let $\beta:\, V_6\rightarrow V_6^s$ denote the blow-down that removes the exceptional divisors associated with $I_0^*$-singularities, and $\pi:\, V_6^s\rightarrow \widetilde{V}_4$ the projection that forgets the $\Proj^2_{[2,3,1]}$-fiber. By construction the divisor $D_B$ on $V_6$ is the pullback of $D_B$ on $\widetilde{V}_4$ under the map $\pi \circ \beta$. Since pullbacks of Cartier and nef divisors under toric morphisms are Cartier and nef (see \cite{cox2011toric} \S$6$),  condition (a) implies that $D_B$ is indeed Cartier and nef on $V_6$.

It remains to be shown that the Weierstrass divisor $D_W$ is also Cartier and nef. To see this, it is enough to show that for every full-dimensional cone $\sigma$ of $V_6$ there exists a global section $s_{\sigma}$ that appears in the Weierstrass polynomial, which does not vanish along $\{x_i=0\,,\forall i\in\sigma(1)\}$. This simply constructs the Cartier data $m_\sigma\in \Delta_{D_W}$ required by Lemma \ref{lem:NP_Cartier_and_nef_is_lattice} and Definition \ref{def:Weil_and_Cartier_divisor}. For this we make a case distinction: first, suppose that $\sigma$ does not contain a generator associated with an NHC or one of the exceptional divisors that resolve $I_0^*$-singularities. Then, the required non-vanishing monomial is either $x^3$, $y^2$ or $\tilde{s}z^6$ --- corresponding respectively to the fiber-loci $y=z=0$, $x=z=0$, or $x=y=0$ ---  with $\tilde{s}$ a non-vanishing section of $D_{g'}$ (which exists by virtue of (b)). Second, suppose that $\sigma$ does contain an NHC or at least one of the exceptional divisors as a generator. Then, by virtue of (c), $\sigma$ involves at most a single NHC (and its resolution divisors). The relevant loci and associated non-vanishing global sections are given in Table \ref{tab:global_sections_NHCs}. $\,\,\Box$
\begin{table}[t]
    \centering
    \begin{tabular}{|c|c| c| c| c | c|}
    \hline
        Vanishing coordinates &  $(x_i,x,e_2)$ & $(x_i,y,e_1)$ & $(x_i,e_1,e_2)$ & $(x,y,e_1)$ & $(x,e_1,e_2)$\\\hline
        Global section & $y^2e_1$ & $x^3$ & $x^3$ & $\tilde{s}x_i^3 z^6$ & $\tilde{s}x_i^3 z^6$\\\hline
    \end{tabular}
    \caption{Local monomial sections that yield Cartier data for the Weierstrass divisor $D_W$ on cones containing an NHC ray or its $I_0^*$-resolution rays. Each listed triple denotes the local part of a maximal cone; the remaining generators are inherited from the base cone. The section $\tilde s$ is chosen to be a non-vanishing monomial section of $D_{g'}$ on the corresponding residual base locus.}
    \label{tab:global_sections_NHCs}
\end{table}

As a consequence of Lemma \ref{lem:uplift_is_refl}, a generic complete intersection
\begin{equation}
    Y_4=D_B\cap D_W\subset V_6
\end{equation}
is smooth away from the singular locus of $V_6$. The remaining singularities are described in the paragraph below.
\end{lemma}

Under the assumptions of Lemma \ref{lem:uplift_is_refl} the anticanonical divisor of $V_6$, 
\begin{equation}
    \overline{K}_{V_6}=D_B+D_W\, ,
\end{equation}
is also nef and Cartier. As a consequence, the convex hull over the primitive generators of the rays of the toric fan of $V_6$ forms a $6d$ reflexive polytope $\Delta_6$. Moreover, the toric fan of $V_6$ is defined from a regular star subdivision of $\Delta_6$. If this isn't already a fine regular star triangulation, one may always refine it by adding all the points in the boundary of $\Delta_6$ and obtaining a regular triangulation. This defines a maximal projective crepant partial (MPCP) resolution of the toric variety \cite{Batyrev:1993oya}. The remaining singularities of $V_6$ are terminal, lie at codimension four in $V_6$, and therefore intersect the Calabi-Yau fourfold at most in isolated points. As we will now show, locally they are orbifold singularities of the form $\mathbb{C}^4/{\mathbb{Z}_2}$ \cite{Morrison1984,Batyrev:1993oya,batyrev1994dualconesmirrorsymmetry}. These are naturally interpreted as the locations of O3-planes, \emph{cf.} \cite{Garcia-Etxebarria:2015wns}.

Indeed, $\mathbb Z_n$-singularities with $n>2$ cannot arise for the following reason. After the MPCP refinement, only four-cones can give singularities intersecting the Calabi--Yau fourfold.
The strata of the higher-dimensional cones do not intersect the CICY and all lower-dimensional cones are smooth.
Singularities on $B_3$ arise from singular cones of $\widetilde\Sigma_4$ which themselves descend from smooth cones of the original FRST $\Sigma_4$. 
Since $[\widetilde N:N]=2$, such cones can at most yield $\mathbb{Z}_2$-singularities in $\widetilde{V}_4$ and hence $B_3$.
The additional fiber directions do not introduce higher-order singularities on the Calabi--Yau: the non-unimodular cones of the $\mathbb P_{[2,3,1]}$ fiber do not intersect the Weierstrass hypersurface, while the cones introduced in resolving the NHCs form a regular triangulation. 
Hence every four-cone whose stratum intersects the Calabi--Yau is either smooth or has index two. 

If, in addition, $D_B$ and $D_W$ admit decompositions in terms of the prime toric divisors that satisfy condition \eqref{eq:partition_condition_2}, then $(D_B,D_W)$ form a nef-partition.

We now want to, starting from a Calabi-Yau orientifold $B_3\subset\widetilde{V}_4$, modify $\widetilde{V}_4$ such that the conditions (a)--(c) of Lemma \ref{lem:uplift_is_refl} are met. 
The modification process can involve the change of triangulation and the blowing up of divisors of the toric fan.
We are specifically interested in finding orientifolds that arise from deforming the original model by only an infinitesimal distance in moduli space.
We will in general not be interested in models that require a large deformation away from the original orientifold. 

The construction of \S\ref{sec:Normal_fan} outlines how to modify $\widetilde{V}_4$ to a toric variety $\widetilde{V}_4'$ in which the base divisor $D_B$ is Cartier and nef, by considering refinements of the normal fan $\Sigma_{\Delta_B}$. 
This ensures that part (a) of Lemma \ref{lem:uplift_is_refl} is satisfied.
We now want to consider a specific implementation of this process that also ensures that the movable part of $6\overline{K}_B$ becomes Cartier and nef. 
From a physics perspective, we furthermore want to ensure that there cannot be divisors along which a generic section of $6\overline{K}_B$ vanishes to an order other than zero or three, as these are the only acceptable generic vanishing orders for a weakly coupled type IIB orientifold.

To begin with, we consider the movable part of $6\overline{K}_{B}$, as defined in (b) of Lemma \ref{lem:uplift_is_refl}.
The integral points in the corresponding Newton Polytope $\Delta_{g'}$ are the same as $\Delta_{6\overline{K}_{B}}\cap M$, but there are no basepoint loci at codimension one: these have been subtracted off. 

To construct a toric variety in which a line bundle with the same sections as $D_{g'}$ is Cartier and nef, we again consider normal fans, continuing our discussion in \S\ref{sec:Normal_fan}. The basic goal is to refine the normal fan constructed in \S\ref{sec:Normal_fan} in ways that guarantee that $D_{g'}$ is Cartier and nef.

First, we define $\Delta_{g'}$ as the \emph{convex hull} over all integral points in the Newton polytope associated with $D_{g'}$. Then, assuming that the Minkowski sum $\Delta_B+\Delta_{g'}$ is a full-dimensional polytope, we may consider its normal fan
\begin{equation}
\Sigma_n:=\Sigma_{\Delta_B+\Delta_{g'}}\, .
\end{equation}

Importantly, by Theorem \ref{thm:normal_fan_of_Minkowski_sum}, $\Sigma_n$ is the coarsest fan for which the divisors whose sections are computed from the integral points in $\Delta_B$ and $\Delta_{D_{g'}}$ are Cartier and nef. We will denote these divisors by $\widehat{D}_B$ and $\widehat{D}_{g'}$ respectively.

Each vertex $m$ of $\Delta_B+\Delta_{g'}$ can uniquely be written as $m=m^{(1)}+m^{(2)}$ with $m^{(1)}\in \Delta_B\,,m^{(2)}\in\Delta_{g'}$ vertices of the two respective polytopes.
Every full-dimensional cone $\sigma_{m}$ of the normal fan is then uniquely  labeled by a vertex $m$ of $\Delta_B+\Delta_{g'}$.

The divisors $\widehat{D}_B,\widehat{D}_{g'}$ have the decompositions
\begin{align}\label{eq:DB_Dg_decompositions}
    \widehat{D}_B\equiv \sum_i b_i D_i=-\sum_i \langle v_i, m_i^{(1)} \rangle D_i\,,\qquad \widehat{D}_{g'} \equiv  \sum_i w_iD_i =-\sum_i \langle v_i, m_i^{(2)} \rangle D_i\, ,
\end{align}
where $m_i=m_i^{(1)}+m_i^{(2)}$ are vertices of $\Delta_B+\Delta_{g'}$ dual to full-dimensional cones of $\Sigma_n$ that contain $v_i$.\footnote{Of course, the $v_i$ are contained in multiple full-dimensional cones, dual to different vertices $m$. In \eqref{eq:DB_Dg_decompositions} it does not matter which one is chosen.}

By virtue of \S\ref{sec:smoothening} the divisor $\widehat{D}_B$ represents the Calabi-Yau orientifold, which we will now refer to as $\widehat{B}$. One now needs to ensure that $6\overline{K}_{\widehat{B}}$ inherits all the desired properties from $D_{g'}$. We have
\begin{align}
    6\overline{K}_{\hat B} = 6(\overline{K}_{\widehat{V}_4}-\widehat{D}_B)=6\sum_i(1-b_i)D_i=6\sum_i (1+\langle v_i,m_i^{(1)}\rangle)D_i\,.\label{eq:6KB_normal_fan}
\end{align}

We would like to refine the normal fan $\Sigma_{n}$, but we first have to ensure the following properties are satisfied in order for the resulting orientifold model to be well-behaved.

The Newton polytope $\Delta_{6\overline{K}_{\widehat B}}$ must contain $\Delta_{g'}$.
Otherwise, one cannot identify the resulting model as a desingularization of the model we started with, as the Weierstrass model would need to be tuned to  a particular sub-locus in its moduli space. 
As such a deformation would not be infinitesimal, we do not want to consider such models in the following. 

This condition can be recast in the form 
\begin{align}
    6(1-b_i)\geq w_i\quad  \Longleftrightarrow \quad \langle v_i,\hat{m}_i\rangle +6\geq 0 \quad \forall i\, , \quad \hat{m}_i:=6m_i^{(1)}+m_i^{(2)}.\label{eq:condition_NP_larger}
\end{align}
If this inequality is saturated for a given $v_i$, the corresponding divisor does \emph{not} host an NHC. If on the other hand
\begin{equation}
    \langle v_i,\hat{m}_i\rangle = -3\, ,
\end{equation}
the divisor $D_i$ hosts an $I_0^*$ singularity. All other cases are incompatible with our weakly coupled type IIB orientifold interpretation, so we really need to impose
\begin{align}
    \langle v_i,\hat{m}_i\rangle\in \{-3,-6\}\, ,
    \label{eq:NHC_3_normal_fan}
\end{align}
for all primitive generators $v_i$ of the normal fan.
\begin{definition}
    Let $\Sigma_\Delta$ be the normal fan constructed from a lattice polytope $\Delta=\Delta_1+\Delta_2$. If all primitive generators $v$ of one-cones $\sigma(1)\in \Sigma_\Delta(1)$ satisfy condition \eqref{eq:NHC_3_normal_fan}, we say that the normal fan $\Sigma_\Delta$ is \emph{admissible}. Otherwise, it is \emph{inadmissible}.
\end{definition}

Given an admissible normal fan $\Sigma_{\Delta_B+\Delta_{g'}}$, it makes sense to proceed and consider refinements of the toric fan. To preserve the nef and Cartier conditions, the most general allowed refinements correspond to subdivisions of $\Sigma_n$, potentially including new rays $v$. 

By the same logic that we applied above to the rays of the normal fan, blow-up rays $v$ must satisfy condition \eqref{eq:NHC_3_normal_fan}, where $\hat{m}_v$ is associated with any full dimensional cone $\sigma$ of the normal fan that $v$ is contained in. Now, by assumption, $v$ can be written as a non-negative linear combination of the generators of $\sigma$, $v=\sum_{v'\in \sigma(1)} a_{v'} v'$, $a_{v'}\geq 0$, and by virtue of \eqref{eq:condition_NP_larger} we have $\langle v',\hat{m}\rangle\leq -3$. Hence, imposing $\langle v,\hat{m}\rangle+6\geq 0$ implies \begin{equation}
    \sum_{v'\in \sigma(1)}a_{v'}\leq 2\, .
\end{equation} 
In particular, the possible blow-up vertices $v$ are actually contained in a \emph{polytope} rather than a cone. Thus, it is natural to consider refinements of the normal fan $\Sigma_n$ that contain \emph{all} such blow-up vertices, and triangulating the resulting vector configuration  (using the tools developed in \cite{MacFadden:2025ssx}).

\begin{definition}
    Let $\Sigma$ be an admissible normal fan constructed from a lattice polytope $\Delta=\Delta_1+\Delta_2$. If $\Sigma'$ is a refinement of the fan $\Sigma$ that contains all primitive generators $v$ that satisfy \eqref{eq:NHC_3_normal_fan}, and forms a regular triangulation of the underlying vector configuration, we call it a \emph{maximal refinement} of $\Sigma$.
\end{definition}

In the resulting toric fourfold, both divisors $\widehat{D}_B$ and $D_{g'}$ are Cartier and nef by construction! After resolving $D_4$-singularities, as explained in \S\ref{sec:smoothening}, the resulting F--theory uplift $Y_4$ is then smooth, up to isolated terminal $\mathbb{Z}_n$-singularities.

One might have attempted a simpler approach: start with the unrefined normal fan $\Sigma_{\Delta_B+\Delta_{6\overline{K}_B}}$, make sure it fulfills the conditions \eqref{eq:condition_NP_larger},\eqref{eq:NHC_3_normal_fan}, and add all the rays of the ``old'' fan $\widetilde{\Sigma}_4$ to it.
This would correspond to the coarsest fan that ensures that parts (a) and (b) of Lemma \ref{lem:uplift_is_refl} are fulfilled.
However, part (c) is guaranteed only for maximal refinements of admissible normal fans.
To prove this, let us assume that at some point in our construction, we have two rays $v_1,v_2$ whose associated divisors host \emph{intersecting} rigid O7 planes. 
The fact that they intersect means that $v_1,v_2$ lie in some common full dimensional cone $\sigma_{m^{(1)}+m^{(2)}}$.
We can resolve the intersection curve by blowing up a (regular) point of the fan by introducing the blow-up ray $v_e=v_1+v_2$ and refining the fan accordingly.
The ray $v_e$ satisfies
\begin{align}\label{eq:blowup_in_max_fan}
    \langle m_2+6m_1,v_e\rangle +6 = \langle m_2+6m_1,v_1\rangle + \langle m_2+6m_1,v_2\rangle +6 = 0\, ,
\end{align}
and hence $v_e$ is automatically included in a maximal refinement.

Here, the final equality holds, because $v_1,v_2$ are rays of NHCs and hence
\begin{align}
    \langle m_2+6m_1,v_{1}\rangle +6 =\langle m_2+6m_1,v_{2}\rangle+6 = 3\, .
\end{align}
We have thus shown that $v_e$ is part of the maximally refined normal fan resolving the intersection of rigid O7 planes. 

Let us summarize the possible failure modes in our construction, i.e., failure to produce divisors with properties listed in Lemma \ref{lem:uplift_is_refl} (a)--(c):
\begin{itemize}
    \item The orientifold hypersurface is irregular, \emph{cf.} Def. \ref{def:regular_orientifold}. In this case either $D_B$ has a divisorial basepoint locus, which the normal fan fails to remove, or the normal fan can simply not be constructed.
    \item The normal fan constructed from $\Delta=\Delta_B+\Delta_{g'}$ is not admissible. This means that either the resulting base space cannot be interpreted as a weakly coupled type IIB orientifold model, or that the model is forced onto a sublocus in its moduli space.
\end{itemize}
If neither of the above situations occur, we may forget about the ``original'' fan $\widetilde{\Sigma}_4$, rename
\begin{align}
    \widehat{V}_4\to \widetilde{V}_4\,,\qquad \widehat{D}_B\to D_B\,,\qquad  \overline{K}_{\widehat{B}}\to \overline{K}_{B}\, ,
\end{align}
and proceed with the F--theory uplift as discussed before in \S\ref{sec:Systematic_F_Theory_uplifts}.
In this case, the assumptions of Lemma \ref{lem:uplift_is_refl} are satisfied. In particular, $D_B$ and $D_W$ are nef Cartier divisors on $V_6$.

However, it frequently happens that the pair $(D_B,D_W)$ does not form a nef-partition, i.e.,  \eqref{eq:partition_condition_2} is not satisfied.

Finally, we address the following natural question: given an admissible normal fan and its maximal refinement, what is the smooth Calabi--Yau variety whose orientifold is given by $\{P_B=0\}\subset \widetilde{V}_4$?
The answer is given by the following Lemma.

\begin{lemma}
    Let $\widetilde\Sigma_4$ and $D_B$ be the maximal refinement of an admissible normal fan and its associated orientifold divisor, as in the construction above. 
    Then, $B_3$ is an orientifold of a Calabi-Yau manifold $Y_3$ which arises as a hypersurface of a toric variety $V_4$ whose fan is an FRST of a reflexive polytope. 
\end{lemma}

\textit{Proof:} 
Recall that the lattice $\widetilde N$ is obtained from the original lattice $N$ as an index-two refinement generated by $\xi/2$. 
Denote by
\begin{align}
    \iota_\xi\,:\qquad N\to \widetilde N\,,
\end{align}
the corresponding inclusion map.
Since $N_\R=\tilde N_\R$, we can consider the cones of $\widetilde{\Sigma}$ as cones in $N_\R$, denote the corresponding fan by $\Sigma_4$ and the associated toric variety by $V_4$.
The inclusion map $\iota_\xi$ then induces a toric morphism
\begin{align}
    f_\xi: V_4\to\widetilde{V}_4\,.
\end{align}
The primitive generators $v_i$ of the rays of $\Sigma_4$ are related to the primitive ray generators $\tilde v_i$ of $\widetilde{\Sigma}_4$ via
\begin{align}
    v_i=\begin{cases}
        \tilde v_i & \tilde v_i\in N\\
        2\tilde v_i & \tilde v_i\notin N
    \end{cases}\,.
\end{align}
Correspondingly, if $D_i$ and $\widetilde D_i$ denote the corresponding prime toric divisors in $V_4$ and $\widetilde V_4$, then
\begin{align}\label{eq:pullback_divisor}
    f_\xi^*\widetilde D_i=\begin{cases}
        D_i & \tilde v_i\in N\\
        2D_i & \tilde v_i\notin N
    \end{cases}\,.
\end{align}
We see from \eqref{eq:O7_class} that $D_B$ can, as a $\mathbb{Q}$-divisor, be written as
\begin{align}
    D_B=\overline{K}_{\widetilde{V}_4}-\frac{1}{2}D_{O7}\,.\label{eq:D_B_Q}
\end{align}
The fact that $\widetilde{\Sigma}_4$ derives from an admissible normal fan means that all $\tilde v_i$ satisfy \eqref{eq:NHC_3_normal_fan}.
As a result, each prime toric divisor $\widetilde D_i$ of $\widetilde V_4$ either hosts one $O7$ plane or it hosts none.
The divisors that do host O7 planes precisely correspond to the primitive ray generators $\tilde v_i$ that satisfy $\tilde v_i\notin N$.
Hence, we can write
\begin{align}
    D_{O7}=\sum_{i:\, \tilde v_i\notin N} \widetilde{D}_i\,.
\end{align}
Inserting this back into \eqref{eq:D_B_Q} yields
\begin{align}
    D_B=\sum_i b_i \widetilde D_i\,,\qquad b_i=\begin{cases}
        1 &  \tilde v_i\in N\\
        \frac{1}{2} &  \tilde v_i\notin N
    \end{cases}\,.
\end{align}
As a result, we find by applying \eqref{eq:pullback_divisor} that
\begin{align}
    f_\xi^*D_B=\sum_i D_i=\overline{K}_{V_4}\,.
\end{align}
Furthermore, $\overline{K}_{V_4}$ is Cartier and nef because these properties are preserved under pullbacks of toric morphisms.
Hence, the fan of $V_4$ arises as a regular and star triangulation of a reflexive polytope. Since we assumed $\widetilde\Sigma$ to be a \emph{maximal} refinement of an admissible normal fan, the triangulation is fine as well. $\,\,\Box$

\addtocontents{toc}{\protect\setcounter{tocdepth}{1}}
\section{Examples}\label{sec:Examples}
Let us now study examples to which we apply the machinery developed above.
All examples can be studied using the software $\mathtt{CYTools}$ \cite{Demirtas:2022hqf} and a corresponding class for Calabi-Yau orientifolds and F--theory uplifts will be added to the public repository of $\mathtt{CYTools}$. In the meantime, our algorithms are available through a GitHub repository  \href{https://github.com/B-Hassfeld/Calabi-Yau-orientifolds-and-F-Theory-uplifts}{\faGithub}.
The following examples have been analyzed using it.

\subsection{A lower-dimensional example: Cubic in $\mathbb{P}^2$}
For simplicity let us begin with a lower-dimensional example.
Let $\Sigma_2$ be the toric fan of $\Proj^2$ with rays given by the columns of 
\begin{align}
    p=\begin{pmatrix}
        1 & 0 & -1\\
        0 & 1 & -1
    \end{pmatrix}\,.
\end{align}
The generic anticanonical hypersurface therein is a $T^2$ cut out by a cubic polynomial in the homogeneous coordinates.
Orbifolding $\Proj^2$ by refining the lattice with the lattice vector $\xi/2=(1,1)/2$ leads to a toric fan $\widetilde{\Sigma}_2$ with primitive rays $v_1,v_2,v_3$ given by the columns of 
\begin{align}
    \tilde{p}=\begin{pmatrix}
        1 & -1 & 0\\
        1 & 1 & -1
    \end{pmatrix}\, ,
\end{align}
expressed in a basis of the refined lattice $\Vec{e}_1=(1,-1)/2$, $\Vec{e}_2=(1,1)/2$.

Both fans are illustrated in Fig.~\ref{fig:P2_fans}.
\begin{figure}
    \centering
    \includegraphics[width=0.7\linewidth]{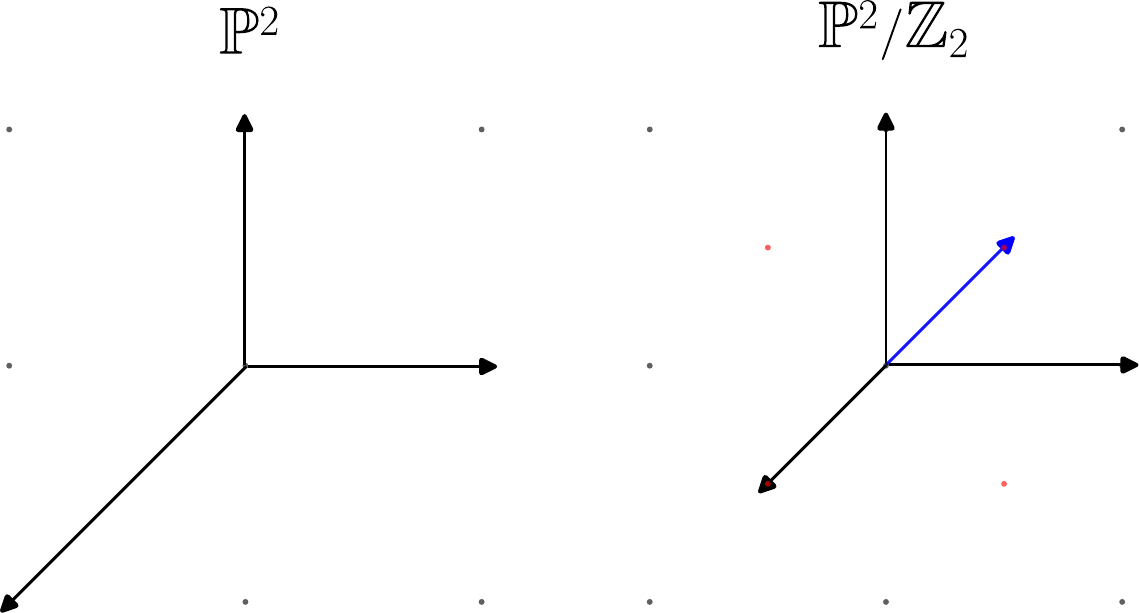}
    \caption{The fan of $\Proj^2$ and of $\Proj^2/\Z_2$ after orbifolding with respect to the vector $\xi=(1,1)$. The red lattice points are due to refinement. The blue ray is the blow-up divisor $v_4$ resolving the $A_1$-singularity.}
    \label{fig:P2_fans}
\end{figure}
The divisor $D_B$ can be represented as
\begin{align}
    D_B=3D_2\,.
\end{align}
The cone spanned by $v_1,v_2$ is singular with singularity type $A_1$. 
As explained in \S\ref{sec:resolution}, we can resolve this singularity by adding the ray $v_4\equiv\frac{1}{2}(v_1+v_2)=(0,1)$ (the blue ray in Fig.~\ref{fig:P2_fans}) to $\widetilde{\Sigma}_2$. 
This ensures smoothness of $\widetilde{V}_2$. 
Alternatively, we can employ the strategy of \S\ref{sec:Normal_fan} and start from the normal fan of $D_B$, which here turns out to be the same as the resolved fan $\widetilde{\Sigma}_2$ above, with rays $v_1,\ldots,v_4$.

In this fan, $D_B$ is represented as
\begin{align}
    D_B=3D_2+D_4\,.
\end{align}
It is straightforward to verify that the normal fan of $\Delta_B+\Delta_{g'}$ is also given by the fan above.
Hence, the construction outlined in \S\ref{sec:Cartier_and_nef} has succeeded and, after uplifting, both $D_B$ and $D_W$ are nef Cartier divisors.

We can now construct the F--theory uplift by following the strategy developed in \S\ref{sec:Systematic_F_Theory_uplifts}.
There is one $I_0^*$ singularity supported along the divisor $D_4\equiv \{\tilde{x}_4=0\}$.\
In total, we find that the rays of the toric fan of the ambient toric variety $V_4$ are given by the columns $v_1,\ldots,v_9$ of 
\begin{align}
    p_4=\begin{pmatrix}
           1 & -1 & 0 & 0 & 0 & 0 & 0 & 0 & 0 \\
           1 & 1 & -1 & 1 & 0 & 0 & 0 & 1 & 2 \\
           0 & 0 & 0 & 0 & 3 & -2 & 0 & -1 & 0 \\
           1 & -2 & 1 &0 & 1 & -1 & 1 & -1 & -1 \\
        \end{pmatrix}\,,
\end{align}
where $v_8,v_9$ are the exceptional divisors that resolve the $I_0^*$ singularity.

It turns out that the divisors $D_B,D_W$ take the form
\begin{align}
    D_B=D_2+D_3\,,\qquad D_W=D_1+\sum_{i=4}^9 D_i\,,
\end{align}
such that the nef-partition property \eqref{eq:partition_condition_2} is also fulfilled. 
For consistency, one can also verify directly that $D_B,D_W$ are indeed Cartier and nef, and, that the four polytopes $\Delta,\Delta^\circ,\nabla,\nabla^\circ$ are reflexive.

The resulting CICY $Y_2$ is a (tuned) family of K3 surfaces. From the type IIB perspective, $12$ of the $16$ D7-branes are at generic positions, while $4$ are on top of one of the four O7-planes.

\subsection{The model \texorpdfstring{$\Proj[1,1,1,1,4]$}{P[1,1,1,1,4]}}
We continue with a model of trilayer type, \emph{cf.} \S\ref{sec:trilayer_polytopes}. This is one of the classic one-parameter Calabi-Yau threefolds that appears in \cite{Morrison:1991cd}.

We begin with the reflexive polytope $\Delta_{\Proj[1,1,1,1,4]}$ whose vertices $v_1,\ldots,v_5$ are the columns of
\begin{align}
    p=\begin{pmatrix}
      -1 & 1 & 1 & 1 & 1\\
      0 & 1 & 0 & 0 & -1\\
      0 & 0 & 1 & 0 & -1\\
      0 & 0 & 0 & 1 & -1
    \end{pmatrix}\,.\label{eq:verts_11114}
\end{align}
This representation of $\Delta_{\Proj[1,1,1,1,4]}$ is in trilayer normal form (see the discussion below Eq.~\eqref{eq:trilayer_grading}) with $v_*=v_1$ and furthermore, the vertices \eqref{eq:verts_11114} are the only nonzero lattice points of $\Delta_{\Proj[1,1,1,1,4]}$.
The unique FRST of this polytope yields a toric fan $\Sigma_4$ of $V_4$ and the anticanonical hypersurface therein is a Calabi-Yau threefold $Y_3$ with Hodge numbers
\begin{align}
    h^{1,1}(Y_3)=1\,,\qquad h^{2,1}(Y_3)=149\,.
\end{align}
The hypersurface equation --- after using symmetries of $V_4$ to remove sections linear in $x_1$ --- is given by
\begin{align}
x_1^2=h(x_i)\,,\label{eq:hyp_EQ}
\end{align}
with $h(x_i)$ a polynomial in the $x_i,\; i>1$.
The polytope $\Delta_{\Proj[1,1,1,1,4]}$ has the simplifying properties that there are no nonzero points not interior to facets in the zero-layer. This property guarantees that the l.h. side of \eqref{eq:hyp_EQ} is the square of the homogeneous coordinate $x_1$. As the orientifold involution is $x_1\mapsto -x_1$, in the toric orbifold $\tilde{x}_1=x_1^2$ is one of the homogeneous coordinates. As a consequence, in the orbifold the hypersurface equation can be solved for $\tilde{x}_1$, 
and the orientifold $B_3$ is a toric variety itself. It is defined by a triangulation of a polytope with vertices
\begin{align}
    p_{\Proj_3}=\begin{pmatrix}
        1 & 0 & 0 & -1\\
        0 & 1 & 0 & -1\\
        0 & 0 & 1 & -1
    \end{pmatrix}\,,\label{eq:B3_P3}
\end{align}
which defines the toric base $B_3=\Proj^3$.

For all trilayer models (as reviewed in \S\ref{sec:trilayer_polytopes}), NHCs are in one-to-one correspondence with points in the ``zero-layer'', other than the origin. Therefore, in this example there are none. 

Following our more general procedure, which we hope will show its worth when we discuss topologically more complex models below, we can also represent $\Proj^3$ as a hypersurface in an orbifold $\widetilde{V}_4$ of $V_4$.
Orbifolding with respect to the lattice refinement $\xi=v_*$ leads to a toric fan whose rays can be represented as the columns of
\begin{align}
    \tilde{p}=\begin{pmatrix}  
      -1 & 2 & 2 & 2 & 2\\
      0 & 1 & 0 & 0 & -1\\
      0 & 0 & 1 & 0 & -1\\
      0 & 0 & 0 & 1 & -1
    \end{pmatrix}\,,\label{eq:orb_11114}
\end{align}
that we also label as $v_1,\ldots,v_5$.
As $\tilde{x}_1\equiv x_1^2$ is a monomial of the hypersurface equation \eqref{eq:hyp_EQ} of $Y_3$ satisfying \eqref{eq:projected_in}, we recover our toric base $B_3=\Proj^3$ as the divisor
\begin{align}
    D_B \equiv D_1 \subset \widetilde{V}_4\,.
\end{align}

We now have two ways of constructing the F--theory uplift.
We can construct appropriately twisted $\Proj_{[2,3,1]}$ fibrations over the toric varieties defined by \eqref{eq:B3_P3} or \eqref{eq:orb_11114} to construct varieties $V_5$ or $V_6$. The Calabi-Yau fourfold $Y_4$ is then constructed as a hypersurface in $V_5$ or as a codimension-two CICY in $V_6$.

The toric variety $V_5$ can be defined as a triangulation of the reflexive polytope $\Delta_5$ with vertices given by the columns of 
\begin{align}
    p_5=\begin{pmatrix}
        1 & 0 & 0 & -1 & 0 & 0\\
        0 & 1 & 0 & -1 & 0 & 0\\
        0 & 0 & 1 & -1 & 0 & 0\\
        0 & 0 & 0 & 0 & 3 & -2\\
        1 & 1 & 1 & 1 & 1 & -1
    \end{pmatrix}\,.\label{eq:p5}
\end{align}
Since $\Delta_5$ is reflexive, it is straightforward to construct $Y_4$ (and also its mirror-dual $X_4$) as anticanonical hypersurfaces and compute the Hodge numbers
\begin{align}
    h^{1,1}(Y_4)=2\,,\qquad h^{2,1}(Y_4)=0\,,\qquad h^{3,1}(Y_4)=3878\,,\qquad \frac{\chi(Y_4)}{24}=972\, ,\label{eq:Hodge_11114}
\end{align}
using the combinatorial formulas of \cite{Klemm:1996ts,Denef:2008wq}.

We can also consider the toric sixfold $V_6$ whose toric fan is constructed via the procedure explained in \S\ref{sec:Weierstrass_model}. The convex hull over the primitive generators of the fan is a reflexive polytope. The fan has a total of eight rays, five inherited from \eqref{eq:orb_11114} and $v_x,v_y,v_z$. No NHCs appear in this model. 

The two divisors defining the codimension-two Calabi-Yau fourfold are given by
\begin{align}
    D_B=D_1\,,\qquad D_W=3D_x\,. 
\end{align}
But, it turns out that $D_W$ can also be represented, by adding a suitable principal divisor, as
\begin{align}
    D_W=\sum_{i=2}^5D_i + D_x+D_y+D_z\,.
\end{align}
The two line bundles are nef, and also form a nef-partition, which is verified by finding an appropriate principal divisor such that $D_W,D_B$ fulfill the condition \eqref{eq:nef_partition_condition}.
We can hence compute the Hodge numbers using \eqref{eq:Hodge_numbers}, and this matches our earlier result \eqref{eq:Hodge_11114}.
One of the two K\"ahler moduli of $Y_4$ derives from the base $B_3$ and the other one comes from the section of the elliptic fiber.
\subsection{The model \texorpdfstring{$\Proj[1,1,1,6,9]$}{P[1,1,1,6,9]}}
Our next example is another trilayer model: the anticanonical hypersurface $\Proj[1,1,1,6,9]$, which appears in \cite{Candelas:1994hw}, and whose F--theory uplift has previously been discussed in \cite{Jefferson:2022ssj}. 

We start with the trilayer polytope $\Delta_{1,1,1,6,9}$ with vertices (in trilayer normal form) given by the columns of
\begin{align}
    p=\begin{pmatrix}
        -1 & 1 & 1 & 1 & 1\\
        0 & 0 & 1 & 0 & -1\\
        0 & 1 & 0 & 0 & -6\\
        0 & 1 & 1 & 1 & -8
    \end{pmatrix}\,.
\end{align}
It again has the feature that there are no nonzero points in the zero-layer that are not interior to a facet, so the orbifold itself is a toric variety.
The anticanonical hypersurface of the toric variety obtained from triangulating $\Delta_{\Proj[1,1,1,6,9]}$ is a Calabi-Yau threefold with Hodge numbers $h^{1,1}=2$ and $h^{2,1}=272$.
There is one NHC in this model, corresponding to the ray generator
\begin{align}
    v_{\text{NHC}} = (1,0,-1,-1)\, ,
\end{align}
which is interior to the one-layer facet.

After orientifolding and using the construction in \S\ref{sec:Weierstrass_model} to lift to the toric sixfold $V_6^s$, we find that the rays $v_1,\hdots ,v_6,v_x,v_y,v_z$ of the toric fan can be represented as the columns of
\begin{align}
     p_6 =\begin{pmatrix}
        -1 & 2 & 2 & 2 & 2& 1& 0 & 0 & 0\\
        0 & 0 & 1 & 0 & -1& 0&0 & 0 & 0\\
        0 & 1 & 0 & 0 & -6 & -1 &0 & 0 & 0\\
        0 & 1 & 1 & 1 & -8 & -1 &0 & 0 & 0\\
        0 & 0 & 0 & 0 & 0 & 0 & 3 & -2 & 0\\
        0 & 1 & 1 & 1 & 1 & 1 & 1 & -1 & 1
    \end{pmatrix}\,.
\end{align}
We see that $v_{\rm NHC}$ lifts to the ray $v_6$ in the uplift.
As before, due to the absence of zero-layer points not interior to facets, $x_1^2$ is a monomial of the Calabi-Yau hypersurface equation, leading to the divisor $D_B$ taking the form
\begin{align}\label{eq:DB_11169}
    D_B=D_1\,.
\end{align}
To resolve the $D_4$ singularity induced by the NHC, we introduce the two additional rays
\begin{align}
    e_1=\begin{pmatrix}
        1 &  0 & -1 & -1 & -1 &  0
    \end{pmatrix}\,,\qquad e_2=\begin{pmatrix}
        2 &  0 & -2 & -2 & 0 &  1
    \end{pmatrix}\,,
\end{align}
so altogether we find that the toric fan of $V_6$ has $11$ rays.

One can now check that the Weierstrass divisor, which is always representable by $D_W=3D_x$, can also be decomposed as 
\begin{align}\label{eq:DW_11169}
D_W=D_x+D_y+D_z+\sum_{i=2}^8D_i\,.
\end{align}
Furthermore, it turns out that the line bundles associated with the divisors $D_B$ and $D_W$ are Cartier and nef starting from any triangulation of the reflexive polytope $\Delta_{1,1,1,6,9}$. Inspecting \eqref{eq:DB_11169} and \eqref{eq:DW_11169} we see that the condition \eqref{eq:partition_condition_2} is satisfied, so we have a nef-partition at hand. In particular, all four polytopes $\nabla,\Delta^\circ,\Delta,\nabla^\circ$ can be defined and are reflexive.

As a consequence, we can use the formula \eqref{eq:Hodge_numbers} to compute the Hodge numbers of the uplift. We find
\begin{align}
    h^{1,1}(Y_4) = 7\,,\qquad h^{2,1}(Y_4)=0\,,\qquad h^{3,1}(Y_4)=7341\,,\qquad \frac{\chi(Y_4)}{24}=1839\label{eq:Hodge_11169}
\end{align}
in accordance with \cite{Jefferson:2022ssj}.

The Hodge number $h^{1,1}$ can be understood in a straightforward manner: two K\"ahler moduli are inherited from the base $B_3$, another one comes from the section of the elliptic fibration, and finally there arise four exceptional divisors corresponding to the resolution of the $D_4$ singularity associated with the single NHC.

Since $\Delta_{1,1,1,6,9}$ has no zero-layer points not interior to facets, the base $B_3$ can again be represented as a toric variety. 
The corresponding rays of the fan are generated by the points in the one-layer of $\Delta_{1,1,1,6,9}$ with the new origin being the point $(1,0,0,0)$.
It is straightforward to construct a $\Proj_{[2,3,1]}$ fibration over this base space, obtain the anticanonical hypersurface and resolve the NHC's $D_4$ singularity by the same method used before. 
The corresponding fan can again be obtained by an FRST of a reflexive polytope such that the Hodge numbers can easily be calculated. 
As expected, they agree with \eqref{eq:Hodge_11169}.

\subsection{The mirror of \texorpdfstring{$\Proj[1,1,1,1,4]$}{P[1,1,1,1,4]}}
Our next example is the \emph{mirror} of $\Proj[1,1,1,1,4]$, another trilayer model with the vertices of the polytope $\Delta_{[1,1,1,1,4]}^\circ$, which is the reflexive dual of $\Delta_{\Proj[1,1,1,1,4]}$ discussed before. Its vertices are the columns of
\begin{align}
    p^\vee_{[1,1,1,1,4]}=\begin{pmatrix}
        -1 & 1 & 1 & 1 & 1\\
        0 & -2 & -2 & -2 & 6\\
        0 & -2 & -2 & 6 & -2\\
        0 & -2 & 6 & -2 & -2
    \end{pmatrix}\,.
\end{align}
The Hodge numbers of the associated Calabi-Yau threefold are $h^{2,1}=1$ and $h^{1,1}=149$, so this is our first example with large $h^{1,1}$.

The polytope $\Delta_{[1,1,1,1,4]}^\circ$ has $153$ nonzero points not interior to facets.
The one-layer facet has $35$ interior points, signaling the presence of $34$ NHCs. 
After orientifolding, i.e., refining the lattice by one-half of $\xi=(-1,0,0,0)$, we therefore obtain a model with 34 NHCs that can be resolved according to the procedure explained in \S\ref{sec:smoothening}.

The resulting F--theory uplift again turns out to be a nef-partition and the Hodge numbers of the elliptically fibered fourfold $Y_4$ are given by
\begin{align}
     h^{1,1}(Y_4)=290\,,\qquad h^{2,1}(Y_4)=0\,,\qquad h^{3,1}(Y_4)=6\,,\qquad \frac{\chi(Y_4)}{24}=76\,.
\end{align}
The divisor of the single non-rigid O7 plane, associated with the ray $(-1,0,0,0)$, is a K3 surface. It has a single deformation modulus: $h^{2,0}(\text{K3})=1$. Adding four D7-branes on the O7 plane, which cancels the D7-tadpole, hence leads to four complex seven-brane moduli.

Together with the axio-dilaton and the original complex structure modulus of the base, they form the six complex structure moduli of the uplift encountered above.

We note that this is the minimal number of complex structure moduli one can hope to find in an F--theory uplift with non-rigid O7 planes (unless the original Calabi-Yau orientifold is rigid, $h^{2,1}_-=0$).
It thus presents itself as a wonderful candidate to explicitly study the moduli stabilization procedure for seven-branes in F--theory.

The Hodge number $h^{1,1}(Y_4)=290$ can also be understood in terms of the K\"ahler moduli of the base, the resolution divisors of the NHCs, the section of the elliptic fibration and four more divisors that arise only in the special situation where the non-rigid O7 plane sits on a K3 surface: the seven-brane gauge theory on K3 has a single adjoint chiral multiplet, so at generic points on the Higgs branch $\mathfrak{so}(8)\rightarrow \mathfrak{u}(1)^4$, leading to four Coulomb branch deformations in $3d$. For seven-branes on a non-rigid divisor $D$ with $h^{2,0}(D)>1$ instead the gauge group breaks to nothing at generic points along the Higgs branch. 

The zero-layer of the polytope $\Delta_{[1,1,1,1,4]}^\circ$ has 22 points such that the orientifold $B_3$ is itself not a toric variety. 
In this example, the only apparent option is thus to describe the F--theory uplift $Y_4$ as a codimension-two CICY in a $6d$ toric variety. 

\subsection{Uplifting the threefold with largest $h^{1,1}=491$}
Let us now consider the class of Calabi-Yau manifolds with largest known $h^{1,1}$. It is constructed from a reflexive polytope $\Delta_{491}$ which satisfies the trilayer criteria \S\ref{sec:trilayer_polytopes}. Its vertices, in trilayer normal form, are the columns of
\begin{align}
    p_{491}=\begin{pmatrix}
  -1 & 1 & 1 & 1 & 1 \\
  0 & -56 & 0 & 1 & 28 \\
  0 & -48 & 1 & 0 & 36 \\
  0 & -42 & 0 & 0 & 42 \\
\end{pmatrix}\,,
\end{align}
and the anticanonical hypersurface $Y_3$ therein has Hodge numbers
\begin{align}
    h^{1,1}(Y_3)=491\,,\qquad h^{2,1}(Y_3)=11\,.
\end{align}
The toric fans $\Sigma_4$, built from triangulations of $\Delta_{491}$, as well as the toric fans $\widetilde{\Sigma}_4$ of the toric orbifolds $\widetilde{V}_4$ have $679$ rays.
The F--theory uplift $Y_4$ has $117$ NHCs and, after crepant resolution of codimension-one singularities in $Y_4$, the toric fan $\Sigma_6$ has $918$ rays. 
The uplift turns out to correspond to a nef partition of a $6d$ reflexive polytope, and its Hodge numbers can be computed to be
\begin{align}
    h^{1,1}(Y_4)=960\,,\qquad h^{2,1}(Y_4)=0\,,\qquad h^{3,1}(Y_4)=264\,,\qquad \frac{\chi(Y_4)}{24}=308\,.
\end{align}
We may again interpret these from the type IIB orientifold perspective: we get $491$ K\"ahler moduli from the Calabi-Yau threefold, one for the elliptic fiber, and $4\times 117 = 468$ from the resolution parameters of the NHCs. The Hodge number $h^{3,1}(Y_4)$ is similarly decomposed into $11$ complex structure moduli of the underlying Calabi-Yau threefold, one for the elliptic fiber, and $252$ seven-brane moduli. These arise from the ``Whitney umbrella''  D7-brane \cite{Collinucci:2008pf} that cancels the seven-brane charge of the orientifold seven-plane wrapped on the divisor $D_1$, which has $h^{2,0}(D_1)=6$. We note that the D7-brane moduli space is \emph{not} the Higgs branch of the $\mathfrak{so}(8)$ gauge theory that arises from wrapping four D7-branes on this divisor. Indeed, the latter has complex dimension $140<252$. Instead, the brane recombination process that corresponds to the $\mathfrak{so}(8)$ Higgs branch generally involves worldvolume fluxes that freeze some of the D7-brane deformations \cite{Gaiotto:2005rp}. The geometric model with $h^{3,1}=264$ is related to this gauge theory Higgs-branch via a non-trivial brane-flux transition.

\subsection{The normal fan at work}
Let us now discuss an example that leads to a nef-partition \emph{only} upon adopting the normal fan construction outlined in \S\ref{sec:Cartier_and_nef}.
The columns of 
\begin{align}\label{eq:p_p_tilde}
    p = \begin{pmatrix}
        -4 & 0 & 0  & 0 & 1 & -2\\
        -3 & 0 & 0 & 1 & 0 & -1\\
        -2 & 0 & 1 & 0 & 0 & -1\\
        -2 & 1 & 0 & 0 & 0 & -1
    \end{pmatrix}\,,\qquad \tilde{p}=\begin{pmatrix}
        -8 & 0 & 0  & 0 & 1 & -4\\
        -3 & 0 & 0 & 1 & 0 & -1\\
        -2 & 0 & 1 & 0 & 0 & -1\\
        -2 & 1 & 0 & 0 & 0 & -1
    \end{pmatrix}
\end{align}
define the rays of a toric fan $\Sigma_4$ --- the lattice points of a reflexive polytope $\Delta$ --- as well as its orbifold $\widetilde{\Sigma}_4$ with respect to the refinement $\xi=(1,0,0,0)$.
The Calabi-Yau threefold constructed as the anticanonical hypersurface in $V_4$ has Hodge numbers 
\begin{align}
    h^{1,1}(Y_3)=2\,,\qquad h^{2,1}(Y_3)=74\,,
\end{align}
and the Calabi-Yau orientifold is cut out of $\widetilde{V}_4$ by the divisor
\begin{align}
    D_B=12D_1 + 6D_6\,.
\end{align}
The Newton polytope of $D_B$ is not a lattice polytope so, as a consequence of Lemma \ref{lem:NP_Cartier_and_nef_is_lattice}, $D_B$ fails to be Cartier and nef given any triangulation of the one-skeleton $\widetilde{\Sigma}_4(1)$.
Upon constructing the toric sixfold $V_6$ as outlined in \S\ref{sec:Systematic_F_Theory_uplifts}, we find the expected result that $D_B$ and $D_W$ do not form a nef-partition. 
There are no NHCs in this model, and one can verify that $D_B$ has a base locus given by
\begin{align}
    BL_{D_B}=\{x_2=x_3=x_6=0\}\, .
\end{align}
The generic hypersurface $D_B\subset \widetilde{V}_4$, as well as the generic orientifold-invariant Calabi-Yau hypersurface in $V_4$, has two isolated singularities along $BL_{D_B}$,
\begin{equation}
    SL_{D_B}=\{x_1=x_2=x_3=x_6=0\}\cup \{x_2=x_3=x_4=x_6=0\}\subset BL_{D_B}\, ,
\end{equation}
where the Jacobian criterion is violated: the generic section vanishes to order two along these cones. 

Locally,  either singularity is isomorphic to a $\mathbb{Z}_2$-orbifold of a conifold,
\begin{equation}
    \sum_{i=1}^4 z_i^2=0\, ,\quad \Vec{z}\in \mathbb{C}^4/{\mathbb{Z}_2}\, ,
\end{equation}
where the $\mathbb{Z}_2$ orbifold acts diagonally on the four local complex coordinates $z_i$.\footnote{Singularities of this type have been studied, e.g., in \cite{Davies:2009ub}.}

Let us see how the normal fan construction of \S\ref{sec:Cartier_and_nef} takes care of these singularities.
We start with the normal fan of $\Delta_B+\Delta_{g'}$, where we have taken the convex hull over the integral points to define $\Delta_B,\Delta_{g'}$ to ensure that they are indeed lattice polytopes.

Let us denote the normal fan by $\widetilde{\Sigma}_4^{(N,0)}$. Its rays are given by the columns $v_1,\ldots,v_7$ of
\begin{align}
    p_{\rm normal}=\begin{pmatrix}
        -8 & 0 & 0  & 0 & 1 & -6 & -1\\
        -3 & 0 & 0 & 1 & 0 & -2 & 0\\
        -2 & 0 & 1 & 0 & 0 & -1 & 0\\
        -2 & 1 & 0 & 0 & 0 & -1 & 0
    \end{pmatrix}\,,\label{eq:p_normal}
\end{align}
and we already see a difference to \eqref{eq:p_p_tilde}: this fan contains seven instead of six rays. By checking \eqref{eq:NHC_3_normal_fan}, one can verify that this normal fan is admissible.
Maximal refinements of this admissible normal fan are by construction compatible with all the monomial sections of $D_B$ and $6\overline{K}_B$.  As explained in \S\ref{sec:Cartier_and_nef} they are defined by appropriate regular triangulations of the rays $v_1,\ldots,v_8$ with $v_8=(-4,-1,-1,-1)$. In what follows, we  denote by $\widetilde{\Sigma}_4^{(N)}$ any such fan.

The maximally refined normal fan $\widetilde{\Sigma}_4^{(N)}$ contains all rays of our original orbifold toric fourfold, plus two further rays, $v_6,v_7$. These two exceptional divisors are responsible for  resolving the two aforementioned isolated singularities. Let us go into a little more detail about how this works, focusing, without loss of generality, on the singularity at $\{x_2=x_3=x_4=x_6=0\}$. 

We have that $v_7=\frac{1}{2}(v_2+v_3+v_4+v_6)$, so the exceptional divisor $D_7$ yields a toric resolution of the $\mathbb{C}^4/{\mathbb{Z}_2}$ quotient singularity in $\widetilde{V}_4$. At the same time, it removes the locus where the Jacobian criterion is violated, yielding a crepant resolution of the singularity in the threefold hypersurface. 

The exceptional divisor in the threefold is a $\mathbb{P}^1\times \mathbb{P}^1$-surface, so in contrast to the small resolution of the ordinary conifold \cite{Candelas:1989js}, here, the resolution variety is divisorial. This blow-up makes sense both in the generic Calabi-Yau orientifold hypersurface $B_3$ as well as in the underlying tuned Calabi-Yau hypersurface $Y_3$. In the latter case, one arrives, on the resolved side, at a new toric fan $\Sigma'$ whose rays form another $4d$ reflexive polytope $\Delta'$. 

As an aside, the toric fan $\Sigma'$ does \emph{not} correspond to an FRST of $\Delta'$, but rather to a \emph{vex triangulation} \cite{MacFadden:2025ssx}. The resolved Calabi-Yau hypersurface is nevertheless smooth. Upon passing through an appropriate flop transition, one reaches a birationally equivalent Calabi-Yau threefold that corresponds to an FRST of $\Delta'$. The Calabi-Yau orientifolds that our normal fan construction produces in this example are equivalent to the orientifolds, with lattice refinement $\xi=(1,0,0,0)$, of the Calabi-Yau hypersurfaces constructed from FRSTs of $\Delta'$. In general, however, the construction of \S\ref{sec:Normal_fan} may yield blown-up Calabi-Yau orientifolds whose underlying Calabi-Yau threefolds have no known description as a hypersurface in a toric variety.

Returning to our example at hand, the class of the Calabi-Yau orientifold in the toric variety $\widetilde{V}_4^{(N)}$ constructed from our fan $\widetilde{\Sigma}_4^{(N)}$ is given by
\begin{align}
    D_B=12D_1+8D_6+D_7+6D_8\,.
\end{align}
By construction, this divisor is Cartier and nef in $\widetilde{V}_4^{(N)}$. 

In this new orientifold model, an NHC is wrapped on $D_7$ and the corresponding singularity has to be resolved when constructing the uplift.
This is easily done by adding two additional rays to the uplift constructed from \eqref{eq:p_normal}.
In total, we find a nef-partition that allows us to compute the Hodge numbers of the uplift $Y_4$:
\begin{align}
    h^{1,1}(Y_4)=9\,,\qquad h^{2,1}(Y_4)=17\,,\qquad h^{3,1}(Y_4)=536\,,\qquad \frac{\chi(Y_4)}{24}=134\,.
\end{align}
The Hodge number $h^{1,1}(Y_4)$ can now be interpreted as follows: two K\"ahler moduli are inherited from the Calabi-Yau threefold, one is added for the elliptic fiber, four arise from resolving the $D_4$ singularity associated with the NHC, and two further K\"ahler deformations arise from the two blow-up divisors associated with the points $v_{7,8}$.

\subsection{Where the normal fan fails}
The next example presents a situation where even the construction of \S\ref{sec:Cartier_and_nef} fails to produce Cartier and nef divisors $D_W,D_B$.

We begin with the rays $v_1,\ldots,v_5$ of the toric fans of $V_4$ and $\widetilde{V}_4$, with the latter arising from the former by orbifolding with respect to $\xi=(0,0,0,1)$:
\begin{align}
p=\begin{pmatrix}
    -1 & 0 & 0 & 0 & 1\\
    -1 & 0 & 0 & 1 & 0\\
    -1 & 0 & 1 & 0 & 0\\
    -1 & 1 & 0 & 0 & 0
\end{pmatrix}\,,\qquad
    \tilde{p}=\begin{pmatrix}
    -2 & 0 & 0 & 0 & 1\\
    -1 & 0 & 0 & 1 & 0\\
    -1 & 0 & 1 & 0 & 0\\
    -1 & 1 & 0 & 0 & 0
\end{pmatrix}\,.
\end{align}
The base divisor has the decomposition
\begin{align}
    D_B=5D_2\,,
\end{align}
and is not Cartier: the Cartier data associated to the cone spanned by $\{v_1,v_3,v_4,v_5\}$ is not integral. 
This is precisely the cone that defines the base locus of $D_B$.
The generic section of $D_B$ vanishes linearly in this cone so we suspect that the hypersurface $D_B$ is actually smooth.
Nevertheless, we will not succeed in finding a nef partition uplift. 

We now employ the strategy of \S\ref{sec:Cartier_and_nef} and compute the normal fan $\Sigma_{\Delta_D+\Delta_{g'}}$.
It is built from the columns $v_1,\ldots,v_6$ of
\begin{align}
    p_{\rm normal}=\begin{pmatrix}
    -2 & 0 & 0 & 0 & 1 & -1\\
    -1 & 0 & 0 & 1 & 0 & 0\\
    -1 & 0 & 1 & 0 & 0 & 0\\
    -1 & 1 & 0 & 0 & 0 & 0
\end{pmatrix}\,.
\end{align}
In the resulting toric variety, we have
\begin{align}
    \widehat{D}_B=5D_1+2D_6\,,\qquad \widehat{D}_{6\overline{K}_B} = -24D_1+6D_2+6D_3+6D_4+6D_5-3D_6\,,
\end{align}
which implies
\begin{align}
    6\overline{K}_{\widehat{B}}-\widehat{D}_{6\overline{K}_B}=-3D_6\, .
\end{align}
This violates the condition \eqref{eq:condition_NP_larger}:  $\Delta_{6\overline{K}_{\widehat{B}}}$ is strictly smaller than $\Delta_{g'}$.
The normal fan is thus inadmissible and we are unable to construct an uplift with Cartier and nef divisors $D_B,D_W$.

\subsection{Statistics}

\begin{figure}
    \centering
    \includegraphics[width=0.7\linewidth]{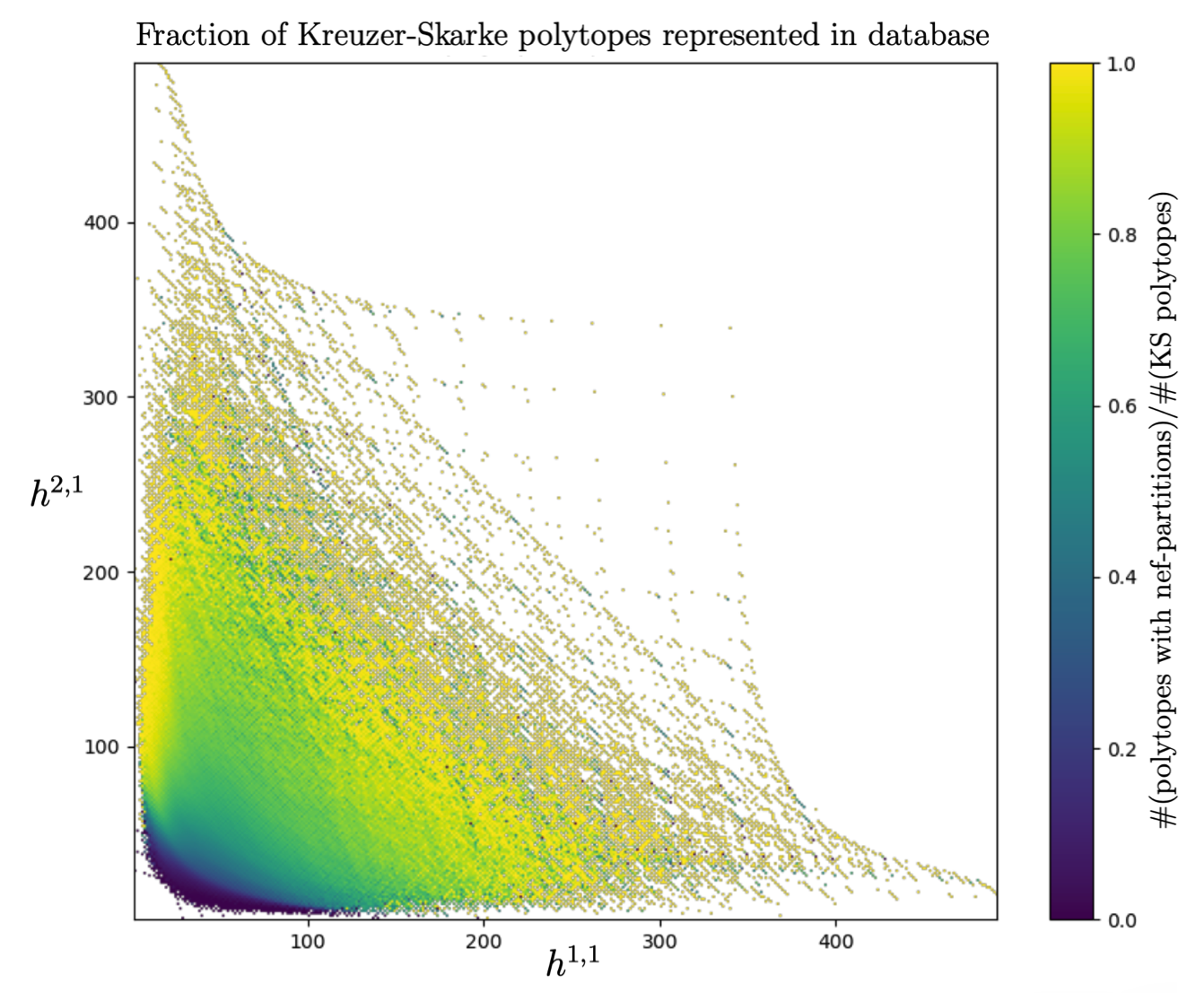}
    \includegraphics[width=0.7\linewidth]{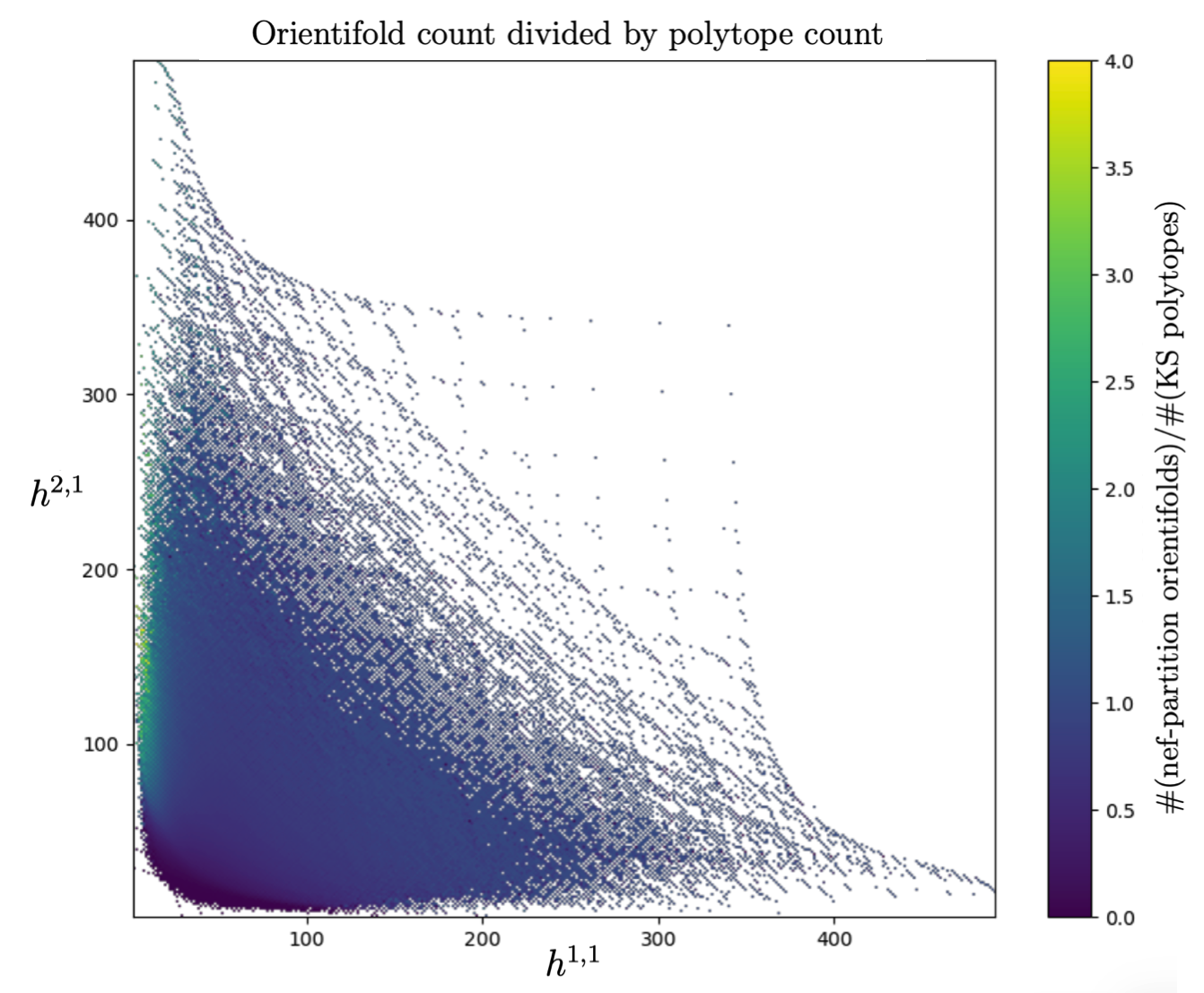}
    \caption{\textbf{Top:} Fraction of polytopes with given Hodge numbers $(h^{1,1},h^{2,1})$ leading to at least one nef-partition uplift.
    \textbf{Bottom:} Number of nef-partition uplifts divided by number of 4d reflexive polytopes for a given pair of Hodge numbers $(h^{1,1},h^{2,1})$.
    The nef-partition density is enhanced at small $h^{1,1}$.}
    \label{fig:statistics}
\end{figure}

Here, we construct explicit ensembles of orientifold models by starting from 4d reflexive polytopes from the Kreuzer-Skarke database \cite{Kreuzer:2000xy}, and applying all possible $\Z_2$ orientifold actions.
Each model then corresponds to a choice of $4d$ reflexive polytope and a lattice point $\xi$ defining the involution.\footnote{Each model may yield many inequivalent but birational configurations associated to different regular triangulations of vector configurations. We only choose one triangulation per polytope.}

In Fig.~\ref{fig:statistics} at the top, the fraction of 4d reflexive polytopes with Hodge numbers $(h^{1,1},h^{2,1})$ leading to at least one nef-partition uplift is displayed. 
At the bottom, we instead show the number of nef-partition uplifts normalized by the number of 4d reflexive polytopes. 
We see that over most Hodge pairs, an $\mathcal{O}(1)$ fraction of Calabi-Yau threefold hypersurfaces have orientifold involutions that lead to nef-partition uplifts.
We note however that the number of nef-partitions is enhanced at small $h^{1,1}$ and fairly suppressed at small $h^{2,1}$ (where $h^{1,1}$ is then necessarily large).
This is due to the fact that the partition property \eqref{eq:partition_condition_2} is easy/hard to fulfill in these regimes.

Figure~\ref{fig:statistics_normal_fan} summarizes the properties of a smaller ensemble of F--theory uplifts generated from Calabi--Yau threefold hypersurfaces with either $1\leq h^{1,1}\leq 9$ or $1\leq h^{2,1}\leq 9$ from the Kreuzer-Skarke list.
Nef-partition uplifts exist most frequently for models with small $h^{1,1}$.

\begin{figure}
        \centering
        \includegraphics[width=0.455\linewidth]{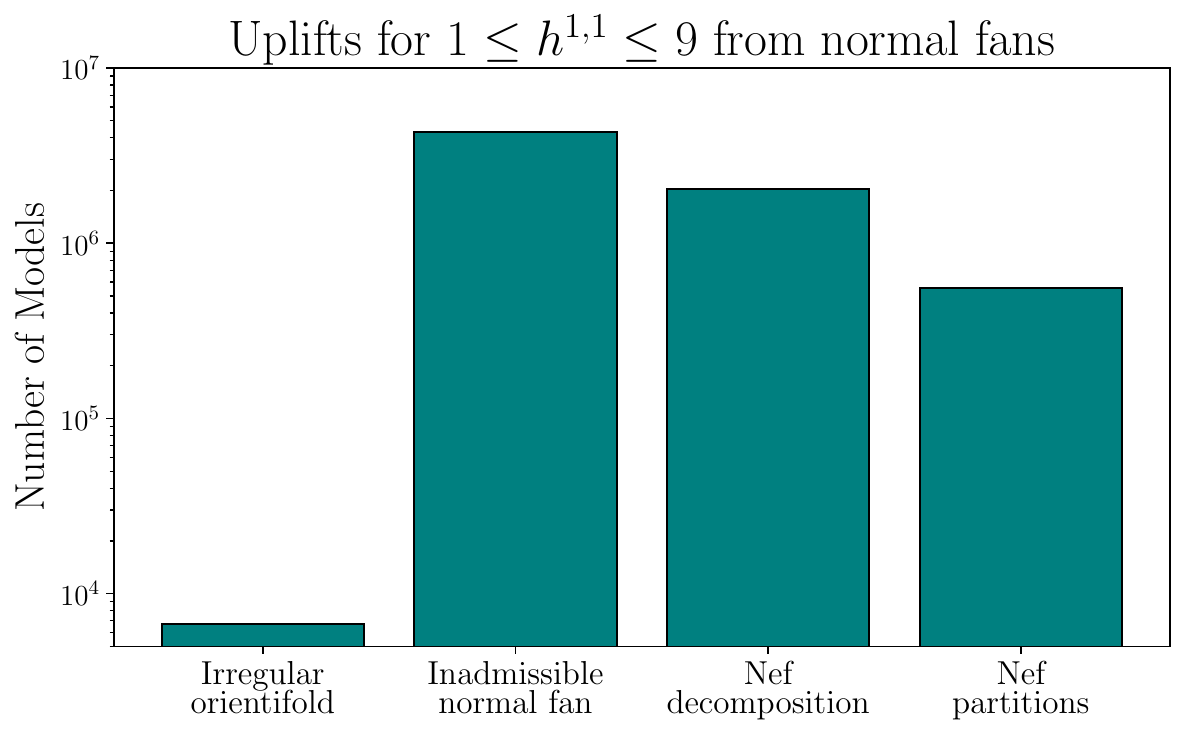}
        \hspace{0.3cm}
        \includegraphics[width=0.455\linewidth]{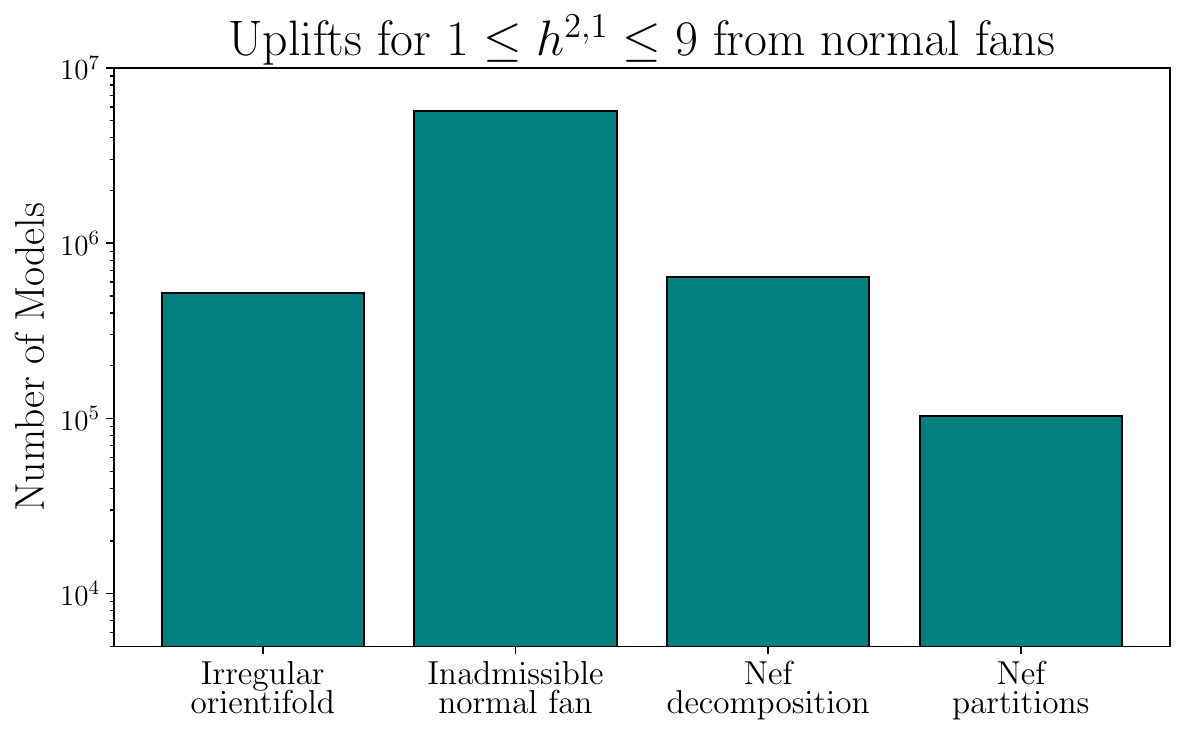}
        \caption{F--theory uplifts generated using the normal fan construction of \S\ref{sec:Cartier_and_nef}, from all polytopes of the Kreuzer-Skarke list with $1\leq h^{1,1}\leq 9$ and $1\leq h^{2,1}\leq 9$.}
        \label{fig:statistics_normal_fan}
\end{figure}

The results of Fig.~\ref{fig:statistics_normal_fan} should be compared to Fig.~\ref{fig:statistics_no_normal_fan}, which shows the same counts, but for uplifts that are constructed directly from Delaunay triangulations of 4d reflexive polytopes, and thus without the normal fan technology of \S\ref{sec:nef_partitions}.
The normal fan construction leads to significantly more nef-decompositions and nef-partitions.

\begin{figure}
        \centering
        \includegraphics[width=0.455\linewidth]{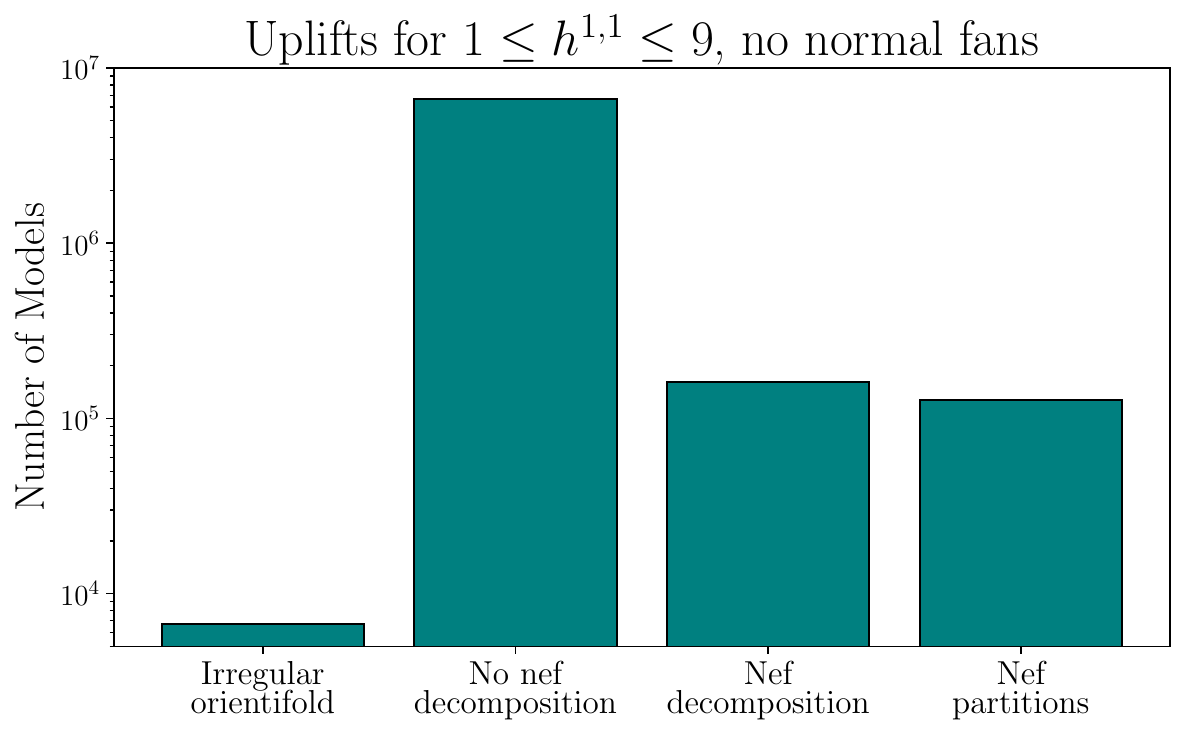}
        \hspace{0.3cm}
        \includegraphics[width=0.455\linewidth]{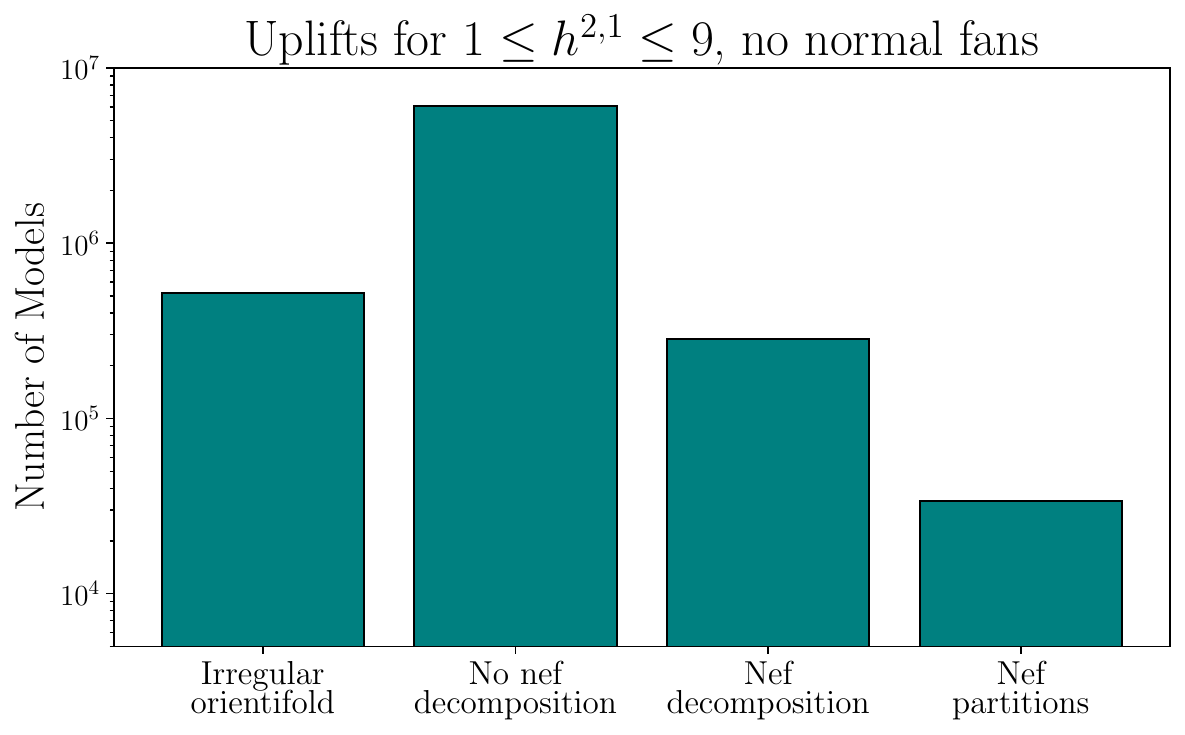}
        \caption{F--theory uplifts from Delaunay triangulations of all  Kreuzer-Skarke polytopes with $1\leq h^{1,1}\leq 9$ and $1\leq h^{2,1}\leq 9$. The normal fan construction of \S\ref{sec:Cartier_and_nef} is not used.}
        \label{fig:statistics_no_normal_fan}
\end{figure}

The data of all our nef-partitions is stored in a new database made available through $\mathtt{huggingface}$ \href{https://huggingface.co/datasets/jakobmoritz/F-theory_Uplifts_Nef-partitions}{\simpleicon{huggingface}}.
In total, we find that $240,136,807$ out of the $473,800,776$ 4d reflexive polytopes admit at least one involution leading to a nef-partition uplift. The overall number of nef-partition uplifts from 4d reflexive polytopes is $255,976,011$.

\section{Summary and Conclusions}
In this work, we have laid the groundwork for evaluating the seven-brane superpotential in type IIB flux compactifications, by developing a combinatorial framework for Calabi-Yau orientifolds and their F--theory uplifts. Applying our methods to the Kreuzer-Skarke list \cite{Kreuzer:2000xy} yields a vast ensemble of pairs of Calabi-Yau threefold orientifolds and their F--theory uplifts.

Our construction proceeds as follows. One starts with a Calabi-Yau hypersurface in a toric variety $Y_3\subset V_4$, for example drawn from the Kreuzer-Skarke list \cite{Kreuzer:2000xy}, and defines a toric orbifold $\widetilde{V}_4$ by refining the lattice underlying the toric fan of $V_4$. A Calabi-Yau orientifold $B_3=Y_3/\Z_2$ results from this as a divisor $D_B\subset \widetilde{V}_4$.

Generically, without further modifications, such orientifolds turn out to be singular, and we have developed an automated combinatorial machinery that resolves the singularities via suitable toric blow-ups. 

Given a Calabi-Yau orientifold defined this way, its F--theory uplift is then constructed as a Weierstrass model:
one first constructs an appropriately twisted $\Proj_{[2,3,1]}$ fibration over $\widetilde{V}_4$ to arrive at a toric sixfold $V_6^s$.
The Weierstrass divisor $D_W$ together with the pullback of $D_B$ to $V_6^s$ then defines a complete intersection Calabi-Yau $Y_4^s=D_W\cap D_B$, the F--theory uplift of $B_3$.

Again, without further modifications, this fourfold is generally singular. The relevant singularities of $Y_4^s$ arise because the elliptic fiber degenerates over the locations of non-Higgsable clusters (NHCs): seven-branes, whose embeddings admit no deformations. We have explained how these remaining singularities can be removed systematically through further toric blow-ups in the ambient variety $V_6^s\rightarrow V_6$. These correspond to moving onto the M-theory Coulomb branch. The resulting fourfold $Y_4$ is the F--theory uplift we are after, and is smooth up to isolated terminal singularities that we interpret as the locations of O3-planes. It arises as a codimension-two hypersurface in the toric sixfold $V_6$, defined by a pair of nef line bundles. The toric ambient sixfold $V_6$ turns out to arise from sufficiently fine regular triangulations of a $6d$ reflexive polytope $\Delta_6$, and our machinery constructs this polytope explicitly.

In many cases, the pair of nef line bundles defining the F--theory uplift satisfies the stronger criterion of a nef-partition, placing such models into the realm where mirror symmetry has been understood by Batyrev and Borisov \cite{borisov1993mirrorsymmetrycalabiyaucomplete,batyrev1994dualconesmirrorsymmetry,Batyrev:1994pg}. For this subset of models our work immediately sets the stage for evaluating the periods of the F--theory uplifts, and thus the seven-brane superpotential in flux compactifications.

We have applied our machinery to a multitude of examples, which highlight the feasibility of our methods even at the largest Hodge numbers. In particular, we have identified various models that arise from nef-partitions, and models with very few complex structure moduli, which form an ideal starting point for investigations into the stabilization of seven-brane moduli in flux compactifications such as \cite{Demirtas:2021nlu,McAllister:2024lnt}. 

It would be very useful to understand the mirror symmetry dictionary in cases where the nef-partition property is not satisfied: such models are perfectly well-behaved compactifications of F--theory, and arise more frequently than nef-partitions. Perhaps, such uplifts can be understood as nef-partitions at higher codimension. 

Our computational tools are available through $\mathtt{CYTools}$ \cite{Demirtas:2022hqf}, and can alternatively be accessed from a GitHub repository \href{https://github.com/B-Hassfeld/Calabi-Yau-orientifolds-and-F-Theory-uplifts}{\faGithub}.  A database of all nef-partition uplifts generated by starting from 4d reflexive polytopes is furthermore available \href{https://huggingface.co/datasets/jakobmoritz/F-theory_Uplifts_Nef-partitions}{\simpleicon{huggingface}}.

\section*{Acknowledgements}  We are indebted to Nate MacFadden and Elijah Sheridan for numerous helpful discussions, and in particular for explaining various details of their work \cite{MacFadden:2025ssx}, as well as the software package $\mathtt{regfans}$ \cite{regfans}, to us. We furthermore thank Elijah Sheridan for useful comments on a draft. We are thankful to the Center for High Throughput Computing (CHTC) at the University of Wisconsin for generously providing computing resources.
Support for this research was provided by the Office of the Vice Chancellor for Research and Graduate Education at the University of Wisconsin-Madison with funding from  the Wisconsin Alumni  Research Foundation.
This work was completed at the Aspen Center for Physics, which is supported by National Science Foundation grant PHY-2210452. We are grateful for its hospitality. BH furthermore acknowledges support from the Simons Foundation (1161654, Troyer). 

\appendix

\newpage

\section{Intersection numbers for non-favorable Calabi-Yau hypersurface embeddings}\label{app:intersections_non_favorable}
In this appendix, we explain how to compute intersection numbers for Calabi-Yau threefold hypersurfaces in toric varieties whose embedding is not favorable.

Let $\Delta^\circ\subset N\simeq  \mathbb{Z}^4$ be a $4d$ reflexive polytope, and let $\Sigma$ be an FRST thereof. Let $\Theta_2^\circ$ be a $2$-face of $\Delta^\circ$ with $k$ points in its interior. Furthermore, let $V=V_{\Sigma}$ be the toric fourfold with toric fan $\Sigma$ and finally let $Y$ be a generic anticanonical hypersurface. To each point $v_I\in\Delta^\circ$ not equal to the origin and not in the interior of a facet, and not in the interior of $\Theta_2^\circ$, we associate a toric coordinate $x_I$, $I=1,\ldots,h^{1,1}(V)+4-k$ and we will denote the toric coordinates associated with the points in the interior of $\Theta^\circ_2$ by $e_\alpha$, $\alpha=1,\ldots,k$. Moreover, we define the toric divisors $D_I:=\{x_I=0\}$, and $E_{\alpha}:=\{e_{\alpha}=0\}$. Finally, we denote by $\{x_i,D_i\}$ the toric coordinates/divisors associated with points $v_i\in \del \theta^\circ_2$. 

First, we consider the intersections of divisors $E_\alpha$ with two or more divisors $\{D_i,E_{\alpha}\}$ in the fourfold $V$, but none of the remaining ones. For this it is useful to note that the face $\Theta_2^\circ$ itself defines a non-compact Calabi-Yau threefold. This is because the points in $\Theta_2^\circ$ plus the origin all lie in a $3d$ hypersurface of $N\simeq \mathbb{Z}^4$ so they can be thought of as points in the lattice $N'\simeq \mathbb{Z}^3$. Furthermore, the points in $\Theta_2^\circ$ lie in a $2d$ affine hypersurface of $N'$. Together with the triangulation induced by $\Sigma$ we thus get a smooth non-compact Calabi-Yau threefold $Y_{loc}$. 

We may now blow down the exceptional divisors $E_{\alpha}$ viewed as divisors in $Y_{loc}$ to obtain a singular toric variety. All the $2d$ cones of its toric fan are still simplicial, but there is now a single (generally) non-simplicial $3d$ cone so there are codimension-three singularities in $Y_{loc}$ that generally are not orbifold singularities. 

Blowing up the exceptional divisors $E_{\alpha}$ (and triangulating) smoothes the threefold $Y_{loc}$ by replacing the singularities at codimension three with the blow-up divisors (and curves). This means that the exceptional divisors $E_{\alpha}$ shrink to a single point $p_s\in Y_{loc}$ upon blowing down. As a consequence, upon viewing the $E_{\alpha}$ as divisors in $V$ they shrink to a singular \textit{curve} $\mathcal{C}$ in $V$, and for small blow-up volumes a tubular neighborhood surrounding $\mathcal{C}$ is isomorphic to $Y'_{loc}\times \mathcal{C}$ where $Y'_{loc}$ is a small tubular neighborhood around $p_s\in Y_{loc}$.

Next, let us denote by
\begin{equation}
\kappa_{\alpha ij}\, ,\quad \kappa_{\alpha\beta i} \, ,\quad \text{and} \quad \kappa_{\alpha\beta \gamma}\, ,
\end{equation}
the triple intersection numbers of $Y_{loc}$ involving at least one exceptional divisor. Then, in $V$ we get $D_i\cap D_j\cap E_{\alpha}\simeq \kappa_{ij\alpha}[\mathcal{C}]\in H^6(V,\mathbb{Z})$, and similarly for the other pairings.

The curve $\mathcal{C}$ now intersects the generic anticanonical hypersurface $Y$ at $k'=\mathcal{C}\cap Y$ \textit{isolated points}, and for small blow-up volumes in small tubular neighborhoods surrounding these points we have $Y\simeq Y_{loc}'$. In particular, the intersections of the exceptional divisors $E_\alpha$ with $Y$ have $k'$ \textit{non-intersecting} irreducible components $E_{\alpha}^a$, $a=1,\ldots,k'$. 

Thus, in the threefold $Y$ we get
\begin{align}
&D_{i}\cap D_{j}\cap E_\alpha^a=\kappa_{ij \alpha}\, ,\quad \forall \, a=1,\ldots,k'\, ,\\
&D_{i}\cap E_{\alpha}^{a}\cap E_{\beta}^{b}=\delta^{ab}\kappa_{i\alpha\beta} \, ,\quad \forall \, a,b=1,\ldots,k'\, ,\\
&E_{\alpha}^{a}\cap E_{\beta}^{b}\cap E_{\gamma}^{c}=\delta^{ab}\delta^{bc}\kappa_{ \alpha\beta\gamma} \, ,\quad \forall \, a,b,c=1,\ldots,k'\, .
\end{align}
Given the intersection of the toric divisors $E_{\alpha}$ with the $\{D_i,E_{\alpha}\}$ and the anticanonical hypersurface $Y$, one immediately obtains the $\kappa_{ij\alpha}$ etc. by dividing the result by $k'$.

All other toric divisors of the ambient variety intersect with the blow-up divisors $E_{\alpha}$ only if they intersect the curve $\mathcal{C}$ in the singular limit. However, these intersection points in the (singular) fourfold do not generically coincide with the Calabi-Yau hypersurface so that corresponding triple intersection numbers in the threefold are all zero.
\section{Intersecting O7 planes and codimension-one uplifts}\label{sec:5d_uplift}

In this section, we briefly discuss models with intersecting O7 planes. 

Consider a trilayer polytope $\Delta$ that has at least one zero-layer point (not equal to the origin) that is not interior to a facet.
An example is given by
\begin{align}
        p=\begin{pmatrix}
      -1 & 1 & 1 & 1 & 1 & 0 \\
      0 & -1 & 0 & 0 & 1 & 0 \\
      0 & -1 & 0 & 1 & 0 & 0 \\
      0 & -2 & 2 & 0 & 0 & 1 \\
    \end{pmatrix}\,,\qquad \tilde{p}=\begin{pmatrix}
  -1 & 2 & 1 & 2 & 2 & 0 & 1 & 2 \\
  0 & -1 & 0 & 0 & 1 & 0 & 0 & 0 \\
  0 & -1 & 0 & 1 & 0 & 0 & 0 & 0 \\
  0 & -2 & 1 & 0 & 0 & 1 & 0 & 1 \\
\end{pmatrix}\,,
\end{align}
where the columns $v_1,\ldots,v_6$ define all points not interior to facets of a trilayer polytope $\Delta$ and thus represent rays of a fan $\Sigma$.

After orbifolding with respect to $v_*=v_1$, we arrive at a fan $\widetilde{\Sigma}_4$ with rays $v_1,\ldots,v_8$ given by the columns of $\tilde{p}$.\footnote{Here, we have resolved two $A_1$-singularities which has led to the addition of two additional rays to the fan.}
For every zero-layer point $v$ that is not interior to a facet, the point $v-\xi$ in $\widetilde{\Sigma}_4$ hosts an NHC.
Hence, there is an NHC along $D_3$.
The divisor $D_B$ can be decomposed as
\begin{align}
    D_B=D_1+D_6\,.
\end{align}

Constructing the F--theory uplift by employing the methods developed in this paper yields a Calabi-Yau fourfold $Y_4$ with Hodge numbers and Euler number given by
\begin{align}
    h^{1,1}(Y_4)=7\,,\qquad h^{2,1}(Y_4)=21\,,\qquad h^{3,1}(Y_4)=1150\,,\qquad \frac{\chi(Y_4)}{24}=286\,.
\end{align}

Upon blowing down the divisor $D_6$, O7 planes along $D_1$ and $D_3$ intersect along a curve in $B_3$. 

This F--theory model can now alternatively be understood as a codimension-one hypersurface in a toric fivefold. This is because after the blow-down, the toric rays associated with zero-layer points have been removed.

The resulting fourfold $Y_4'$ is the anticanonical hypersurface of the 5d toric variety constructed as a triangulation of the reflexive polytope $\Delta_5$ with vertices
\begin{align}
    p_{\Delta_5}=\begin{pmatrix}
  -1 & 0 & 0 & 0 & 0 & 1 \\
  -1 & 0 & 0 & 0 & 1 & 0 \\
  -2 & 0 & 0 & 1 & 0 & 0 \\
  0 & -2 & 3 & 0 & 0 & 0 \\
  1 & -1 & 1 & 1 & 1 & 1 \\
\end{pmatrix}\,.
\end{align}
As $\Delta_5$ is reflexive, the Hodge numbers and the Euler number of $Y_4'$ can be computed using Batyrev's formula for hypersurfaces \cite{batyrev1993dualpolyhedramirrorsymmetry}:
\begin{align}
    h^{1,1}(Y_4')=2\,,\qquad h^{2,1}(Y_4')=0\,,\qquad h^{3,1}(Y_4')=3786\,,\qquad  \frac{\chi(Y_4')}{24}=949\,.
\end{align}

We see that the Hodge numbers of $Y_4'$ --- and the D3-tadpole --- significantly differ from the ones of $Y_4$. The transition between $Y_4$ and $Y_4'$ involves two steps. First, the exceptional divisors associated to the zero-layer points of $\Delta$ are blown-down, leading to transverse intersections of O7-planes. Second, after the blow-down, additional complex structure deformations become available, which generically lead to a recombination of the intersecting O7-planes into a single smooth O7-plane. As the D3-tadpole differs on the two sides of this geometric transition we expect non-trivial $G_4$-flux on the deformation side, as in \cite{Gaiotto:2005rp}.
Furthermore, by analyzing the monodromy polynomial, one finds that the generic gauge algebra on the NHC is $\mathfrak{g}_2$ on the deformation side, rather than $\mathfrak{so}(8)$.

This behavior generalizes to any trilayer polytope with non-trivial zero-layer points:
a codimension-one uplift can always be constructed after blowing down all divisors associated to zero-layer points, and the resulting 5d polytopes are always reflexive.

Furthermore, we were able to verify that many of the codimension-two nef-partitions we found do not allow for a codimension-one description in terms of any known reflexive polytope from the Sch\"oller-Skarke list \cite{Scholler:2018apc}. 
We conclude from this analysis that working with codimension-two CICYs is necessary for dealing with the F--theory uplifts of Calabi-Yau orientifolds derived from the Kreuzer-Skarke list. 
Nevertheless, blowing down the zero-layer points of trilayer polytopes and uplifting to F--theory results in valid Calabi-Yau fourfold models with potentially interesting physics.
Since the underlying $5d$ polytopes are reflexive, mirror symmetry is again understood combinatorially \cite{Batyrev:1993oya}.
\section{Mirror symmetry from nef-partitions}\label{app:nef_partitions}

In this appendix we review the mirror symmetry dictionary of Batyrev and Borisov \cite{borisov1993mirrorsymmetrycalabiyaucomplete,batyrev1994dualconesmirrorsymmetry,Batyrev:1994pg} in more detail than in the main text.

Consider a toric variety $V$ with toric fan $\Sigma$ and $k$ divisors $D^{(1)},\ldots,D^{(k)}$ decomposed in terms of prime toric divisors as
\begin{align}
    D^{(\alpha)} = \sum_{j} a_{j}^{(\alpha)}D_j\,.\label{eq:D_alpha_i_rep}
\end{align}
Each divisor comes with a Newton Polytope $\Delta^{(\alpha)}$ in $M_\R$. 
From the adjunction theorem, we conclude that in order for the complete intersection $Y=D^{(1)}\cap \ldots \cap D^{(k)}$ to be a Calabi-Yau variety, $\sum_{\alpha=1}^k D^{(\alpha)}=\overline{K}_V$ must hold.
By an appropriate shift of the $D^{(\alpha)}$ by principal divisors, one can thus arrange that
\begin{align}
    \sum_{\alpha=1}^k a_{j}^{(\alpha)} = 1 \,,\qquad \forall j\,.\label{eq:a_sum_to_1}
\end{align}
If each divisor $D^{(\alpha)}$ is Cartier and nef, we conclude from Lemma \ref{lem:nef_Mink_sum} that
\begin{align}
    \Delta \equiv \Delta^{(1)} + \ldots + \Delta^{(k)}\label{eq:Delta_M_sum_def}
\end{align}
is the Newton Polytope of $\overline{K}_V$.
Furthermore, as both the Cartier and the nef property are linear, as can directly be seen from the Definitions \ref{def:Weil_and_Cartier_divisor} and \ref{def:nef}, $\overline{K}_V$ is then also Cartier and nef.
This implies that $\Delta$ is a reflexive polytope\footnote{In general, the Newton Polytope is only defined up to a linear shift, as a consequence of Lemma \ref{lem:shift}. However, the condition \eqref{eq:a_sum_to_1} ensures that the resulting $\Delta$ has the origin as its unique interior point, so it is reflexive.} and its dual is the polytope
\begin{align}
    \Delta^\circ=\text{Conv}(\Sigma(1))\,.\label{eq:Delta_N_conv_def}
\end{align}
As explained in \cite{borisov1993mirrorsymmetrycalabiyaucomplete}, mirror symmetry can now be phrased combinatorially,  if in addition, the divisors $D^{(\alpha)}$ can be decomposed as in \eqref{eq:D_alpha_i_rep},\eqref{eq:a_sum_to_1} and satisfy the extra condition
\begin{align}
    a_j^{(\alpha)}\in\{0,1\}\,,\qquad \forall \alpha,j\,.\label{eq:nef_partition_condition}
\end{align}
In this case, the divisors $D^{(\alpha)}$ form a \emph{nef-partition}. 
The terminology comes from the fact that \eqref{eq:nef_partition_condition} partitions the rays $\Sigma(1)$ in $k$-many groups
\begin{align}
    E^{(\alpha)} =\{v_j \in \Sigma(1)\,|\,a_j^{(\alpha)}=1\}\,.
\end{align}

The property \eqref{eq:nef_partition_condition} implies that
\begin{align}
    \Delta':= \text{Conv}(\Delta^{(1)},\ldots,\Delta^{(k)})\label{eq:Delta_M_conv_def}
\end{align}
is reflexive \cite{borisov1993mirrorsymmetrycalabiyaucomplete}. 

To see this, we deduce from \eqref{eq:nef_partition_condition}  and Def.~\ref{def:Newton_Polytope} that $0\in \Delta^{(\alpha)}$ for all $\alpha$.
This implies $0\in \Delta'$ as well.

In fact, the origin has to be an \emph{interior} point of $\Delta'$. Indeed, suppose it were a boundary point. Then, there would have to exist a hyperplane $H$ with the property that all points $m\in \Delta^{(\alpha)}$ satisfy $H\cdot m\geq 0$. This, however, implies that the origin is also a boundary point of $\Delta$, and this contradicts our assumption that this polytope is reflexive.

Next, $\Delta'$ is also a lattice polytope, because the Newton polytopes of Cartier and nef divisors are lattice polytopes. 
It hence remains to be shown that the polar dual of $\Delta'$ is a lattice polytope. 
To demonstrate this one recalls the property \eqref{eq:nef_partition_condition} together with the definition of the polar dual polytope, Def.~\ref{def:reflexive_and_polar_dual}. These guarantee that this polytope is given by
\begin{align}
    \nabla := \nabla^{(1)}+\ldots+\nabla^{(k)}\,,\qquad \nabla^{(\alpha)}:=\text{Conv}(\{0\}\cup E^{(\alpha)})\,.\label{eq:Delta_N_sum}
\end{align}
As all $\nabla^{(\alpha)}$ are lattice polytopes, so is \eqref{eq:Delta_N_sum} and hence $\nabla$ and $\Delta'\equiv \nabla^\circ$ form a reflexive pair.

\begin{figure}
        \centering
        \includegraphics[width=0.8\linewidth]{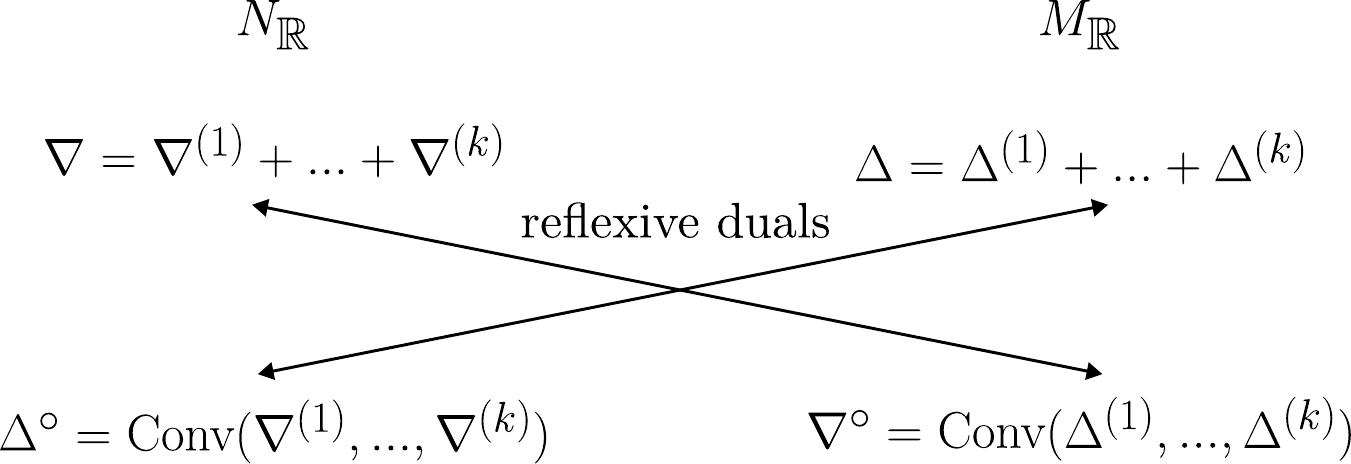}
        \caption{Dualities for Batyrev-Borisov mirror symmetry coming from nef-partitions.}
        \label{fig:app:BB_Mirror_symetry_tot}
\end{figure}

In the situation of a CICY $Y$ being defined by a nef-partition, mirror symmetry can be understood as follows.
An FRST of the polytope $\nabla^\circ$ yields a toric fan of a toric variety $V^\vee$.
The polytopes $\nabla^{(\alpha)}$ define Newton polytopes of Cartier and nef divisors defined on $V^\vee$ that lead to a codimension-$k$ CICY $X$ dual to $Y$.
The duality operations are depicted in Fig.~\ref{fig:app:BB_Mirror_symetry_tot}.
A different approach to mirror symmetry for CICYs is to study the dualities of \textit{reflexive Gorenstein cones}, as explained in \cite{batyrev1994dualconesmirrorsymmetry}.

\section{Hodge numbers for CICYs}\label{app:Hodge_numbers}
Given a CICY defined by a nef-partition $D^{(1)},\ldots,D^{(k)}$, the Hodge numbers can be computed combinatorially, as explained in \cite{Batyrev:1994pg,Batyrev:1995ca}.
In order to do so, one first constructs the \emph{Cayley polytopes} $\Delta^{\rm Cay}$ and $\nabla^{\rm Cay}$ defined by
\begin{align}
    \Delta^{\rm Cay} := \text{Conv}(\Delta^{(1)}\times e_1,\ldots,\Delta^{(k)}\times e_k)\subset M_\R\times \R^k\,,\\
    \nabla^{\rm Cay} := \text{Conv}(\nabla^{(1)}\times e_1,\ldots,\nabla^{(k)}\times e_k)\subset N_\R\times \R^k\,,
\end{align}
with $e_i$ the unit vectors of $\R^k$.

From these, one constructs a polynomial generating functional for the Hodge numbers of the pair of CICYs \cite{Batyrev:1995ca,Batyrev:1994ju,batyrev1998stringyhodgenumbersvarieties}. The Hodge numbers are obtained as coefficients thereof.

For the case of $k=2$, one finds \cite{Novoseltsev2011}
\begin{align}
    \begin{split}
        h^{1,1}(Y)=h^{d-1,1}(X)=\ell(\Delta^{\rm Cay})-2-n - \sum_{y=0}[\ell^*(2\cdot y^\vee)-(n+1)\ell^*(y^\vee)]\\
    +\sum_{\dim y=1}\ell^*(y^\vee)+\sum_{\dim y=1}\ell^*(y)[\ell^*(2\cdot y^\vee)-n\ell^*(y^\vee)]\\
    -\sum_{\dim y=2}[\ell(y)-\ell^*(y)-3]\ell^*(y^\vee)+\sum_{\dim y=3}[\ell^*(2\cdot y)-4\ell^*(y)]\ell^*(y^\vee)\\
    -\sum_{\dim x=2,\dim y=3, x<y}\ell^*(x)\ell^*(y^\vee)\,,
    \end{split}\label{eq:Hodge_numbers}\\
    \begin{split}
        h^{2,1}(Y)=h^{2,1}(X) =     \sum_{\dim y=2}\ell^*(y)[\ell^*(2\cdot y^\vee)-(n-3)\ell^*(y^\vee)]\\
        +\sum_{\dim y=4}\ell^*(y^\vee)[\ell^*(2\cdot y)-5\ell^*(y)]\\
        -\sum_{\dim x=2,\dim y=3,x<y}\ell^*(x)\ell^*(y^\vee)-\sum_{\dim x=3,\dim y=4,x<y}\ell^*(x)\ell^*(y^\vee)\,,
    \end{split}\label{eq:Hodge_number_h21}
\end{align}
where the sums run over faces of $\nabla^{\rm Cay}$. Here,  $\ell(y)$ denotes the number of points in the face $y$, and $\ell^*(y)$ counts the number of points in the interior of $y$. Moreover, $y^\vee$ denotes the dual face of $y$ and $x<y$ denotes that $x$ is a face of $y$.\footnote{The Cayley polytopes themselves are not reflexive polytopes. Given a face $y\subset \nabla^{\rm Cay}$, its dual face $y^\vee\subset \Delta^{\rm Cay}$ is defined by $y^\vee = \{m\in \Delta^{\rm Cay}, \langle m,v\rangle \geq 0 \;|\;\forall v\in y\}$.}

From the independent Hodge numbers of the Calabi-Yau, one computes its Euler number according to
\begin{align}
    \chi(Y)=48 + 6(h^{1,1}(Y)+h^{3,1}(Y)-h^{2,1}(Y))\,.
\end{align}

In principle, the generating functional introduced in \cite{Batyrev:1995ca} may also be used in situations when all divisors are Cartier and nef but the partition property \eqref{eq:nef_partition_condition} does not hold. 
Indeed, in this situation one can still define a reflexive Gorenstein cone and its associated generating functional for the Hodge numbers \cite{Batyrev:2007cq}. 
However, evaluating this function explicitly becomes difficult, if no dual nef-partition exists.
\bibliography{biblio}
\bibliographystyle{JHEP}
\end{document}